\documentclass[sigconf]{acmart}

\usepackage{booktabs}
\usepackage{amsmath, amssymb}	
\usepackage{relsize}			
\usepackage{url}				
\usepackage{amsthm}				
\usepackage{todonotes}			
\usepackage{bytefield}          
\usepackage{subcaption}
\usepackage{caption}
\usepackage{xfrac}
\usepackage{tabularx,ragged2e,booktabs}

\usepackage{graphicx}
\graphicspath { {figures/} }

\newcommand{\set}[1]{\ensuremath{\left\{#1\right\}}}	
\newcommand{\parans}[1]{\ensuremath{\left(#1\right)}}	

\DeclareMathOperator*{\argmax}{arg\,max}

\newcommand{\remove}[1]{}

\newtheorem{lem}{Lemma}
\newtheorem{thm}{Theorem}

\begin{document}
\title{Revisiting Client Puzzles for State Exhaustion Attacks Resilience}
\subtitle{Can Proof-of-Work Actually Work?}

\begin{abstract}
In this paper, we
address the challenges facing the adoption of client puzzles as means to protect
the TCP connection establishment channel from state exhaustion DDoS attacks. We model the problem of selecting
the puzzle difficulties as a Stackelberg game with the server as the leader and the clients as the followers and 
obtain the equilibrium solution for the puzzle difficulty. We then present an implementation of client puzzles 
inside the TCP stack of the Linux 4.13.0 kernel. 
We evaluate the performance of our implementation and the obtained solution against a range of attacks through 
experiments on the DETER testbed. Our results show that client puzzles are effective at boosting 
the tolerance of the TCP handshake channel to state exhaustion DDoS attacks by rate limiting the flood rate 
of malicious attackers while allocating resources for legitimate clients. Our results illustrate the 
benefits that the servers and clients amass from the deployment of TCP client puzzles and incentivize their adoption 
as means to enhance tolerance to multi-vectored DDoS attacks
\end{abstract}



\author{Mohammad A. Noureddine}
\affiliation{%
  \institution{Department of Computer Science \\ University of Illinois at Urbana-Champaign}
  \city{Urbana}
  \state{IL}
}
\email{nouredd2@illinois.edu}

\author{Ahmed Fawaz}
\affiliation{%
  \institution{Department of Electrical and Computer Engineering \\ University of Illinois at
  Urbana-Champaign}
  \city{Urbana}
  \state{IL}
}
\email{afawaz2@illinois.edu}

\author{Tamer Ba\c{s}ar}
\affiliation{%
  \institution{Coordinated Science Laboratory \\ University of Illinois at
  Urbana-Champaign}
  \city{Urbana}
  \state{IL}
}
\email{basar1@illinois.edu}

\author{William H. Sanders}
\affiliation{%
  \institution{Department of Electrical and Computer Engineering \\ University of Illinois at
  Urbana-Champaign}
  \city{Urbana}
  \state{IL}
}
\email{whs@illinois.edu}

\maketitle

\section{Introduction} \label{s:intro}
In recent years, the scale and complexity of Distributed Denial of Service (DDoS) attacks have grown significantly. 
The introduction of DDoS-for-hire services has substantially decreased the cost of launching complex, 
multi-vectored attacks aimed at saturating the bandwidth as well as the state of a victim server, with 
the possibility of inflicting severe damages at lower costs~\cite{anreport2018, dyn2016}. 


Common mitigations to large-scale DDoS attacks are focused around cloud-based protection-as-a-service providers, such as
CloudFlare. 
When under attack, a victim's traffic is redirected to massively
over-provisioned servers where proprietary traffic filtering techniques
are applied and only traffic deemed benign is forwarded to the victim. 
The relative success of such over-provisioning techniques in absorbing {\em volumetric} attacks has pushed attackers to
expand their arsenal of attack vectors to span multiple layers of the OSI network stack~\cite{ddosthreatq417}.
In fact, $39.8\%$ of the attacks launched through the Mirai botnet  were aimed at TCP state exhaustion while
$32.8\%$ were volumetric~\cite{mirai_usenix}; the Mirai source code contained more than 10 vectors in its arsenal 
of attacks~\cite{mirai_source}. 

State exhaustion and application layer attacks are particularly challenging.
Attackers can masquerade such attacks as benign traffic by leveraging a large number of machines that can use their
authentic IP addresses~\cite{anreport2018}, and can thus bypass 
cloud-based protection services, capabilities,
and filtering techniques~\cite{Parno2007,Xiaowei2008,Xiaowei2005,liu2010,liu2008,mahajan2002}. 
This is further exacerbated by 
the imbalance between the cost of launching a
multi-vectored DDoS attack and the cost of mitigating one. Launching an attack incurs an average cost of $\$66$ per
attack and can cause potential damage to the victim of around $\$500$ per minute, not including the costs paid to the
protection service provider~\cite{arbor2017}. 
Facing this hybrid and imbalanced attack landscape, it is essential that we develop mitigation
techniques that can defend against the different attack vectors involved in recent DDoS attacks.

In this paper, we revisit the application of client puzzles as a mechanism to resist state exhaustion DDoS attacks. 
Client puzzles are a promising technique that alleviates the cost imbalance between the attacker and the defender with 
only software-level modifications at the end hosts 
and no changes to the Internet infrastructure~\cite{feng03case, juels1999}.  
At their essence, 
client puzzles attempt to hinder the malicious actors' ability to flood the server at low cost by 
forcing all clients, both benign and malicious,
to solve computational puzzles for each request they make.
%
 
 
While TCP client puzzles provide a promising technique to resist state exhaustion attacks, 
they have not seen their way into adoption due to (1) the lack of guidelines on the difficulty setting,
and (2) the lack of publicly available implementations and performance studies~\cite{wang2003, feng03case, ietf-draft}.  
A TCP client puzzle's difficulty determines the computational burden placed on the server and clients --
an easy puzzle offers better usability at the cost of security while a harder puzzle provides better security but can
deter users. 
The lack of guidelines on the difficulty setting forces administrators to speculate whether to overload their servers or turn users away, leading them 
to overlook puzzles as an effective protection mechanism. 
Additionally, the few existing implementations of TCP client puzzles~\cite{feng2005, wang2003, wang2004, dean2001} are outdated and are not publicly 
available, further hindering the community's ability to evaluate their effectiveness and adopt them.

In this work, we make the following contributions to address the shortcomings in TCP client puzzles research. 
First, we introduce a theory for determining an appropriate TCP puzzle difficulty based on the game-theoretic 
Stackelberg interaction between the defender 
and the clients~\cite{basar2002,shen2007_2,shen2007} (Sections~\ref{s:framework} and~\ref{s:application}). 
Using the theory we established, we provide a practical method to select the TCP puzzles difficulty based on the defender's capabilities 
and the expected computational prowess of the clients.

Then, we design, implement and evaluate an extension to TCP to support client puzzles using our practical difficulty setting method. 
We incorporate puzzles into the TCP handshake and do not interfere with the operation of the protocol otherwise. 
We efficiently encode the 
challenges and their solutions into the {\em options} of the TCP header resulting in a TCP puzzles extension with
low packet-size overhead. 
Then, we implemente TCP puzzles as part of the Linux kernel TCP stack (Section~\ref{s:implementation}). 
To the best of our knowledge, we provide the \emph{first} publicly available implementation of TCP puzzles for the Linux kernel\footnote{The open sourced patch for 
v4.13.0 is available at \url{https://github.com/nouredd2/linux}}. 
Our implementation maintains the statelessness property of client puzzles:
the server relies solely on the status of its internal queues to infer the presence of an attack and creates no state
until a puzzle is verified to be valid.

We evaluated the performance of our TCP puzzles against a range of attacks through 
experiments performed using the DETER testbed (Section~\ref{s:evaluation}).
Our results show the effectiveness of TCP puzzles in boosting tolerance against state exhaustion attacks. First, a server using
TCP client puzzles, configured to use the game-theory-based difficulty setting method, tolerates both SYN and
connection floods that would bring down an unprotected server or one that relies solely on SYN
cookies~\cite{cryptosyncookies}. Furthermore, clients willing to solve the challenges are
almost always able to receive service during an attack. Moreover, TCP puzzles result in negligible overhead for the server, 
while significantly increasing the cost of launching a DDoS attack -- the size of a botnet has to increase by a factor of 200, 
and IoT-based botnets become unable to launch such attacks -- thereby removing the low-cost assets in an attacker's arsenal. 
We believe that TCP client puzzles can be a strong companion to protection-as-a-service solutions to further resist
multi-vectored DDoS attacks.

\section{Background and Related Work}

In this section, we review the TCP three-way handshake and TCP state exhaustion attacks. 
We then present client puzzles and the limitations of the current approaches to using them for TCP state exhaustion attacks
resistance. For the remainder of this paper, we use the terms puzzles and challenges interchangeably.  

\subsection{TCP primer and SYN flood attacks}
In current TCP implementations,
a client initiates a TCP connection by sending a SYN packet to the server. Upon receiving the SYN packet, the server saves state for this new incoming connection request in a data structure, often referred to as the Transmission Control Block (TCB), and then sends a SYN-ACK packet and waits for the client to acknowledge receiving this packet. A half-open connection is one for which the client's ACK packet has not been received yet; those new connection sockets are queued into a \texttt{listen} queue. 
The number of
elements in this queue is upper bounded by an implementation parameter, called
the {\em backlog}, that
bounds the server's memory usage as to avoid having the system's resources exhausted.
Once a connection has been established, the server moves it into the \texttt{accept} 
queue.
%
A socket is removed from the accept queue once the
server's application accepts it for
processing. On the other hand, a half-open connection socket is removed from
the queue if it expires before it receives an
acknowledgment from the client~\cite{rfc793}. 
Once the server's listen or accept queues overflow, 
it either (1) no longer accepts incoming connections or (2) drops old
connection sockets from the
appropriate queue to free space for newer connections. 

\begin{figure}
  \centering
  \begin{subfigure}[t]{0.45\columnwidth}
      \includegraphics[width=1.\textwidth]{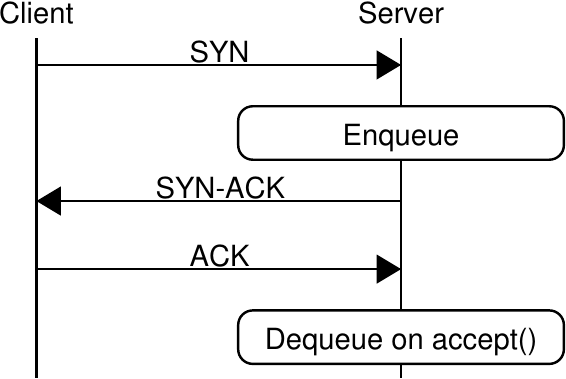}
    \label{fig:syn}
  \vspace{-1em}
    \caption{}
    \vspace{1em}
  \end{subfigure}
  ~~
  \begin{subfigure}[t]{0.45\columnwidth}
    \centering
      \includegraphics[width=1.\textwidth]{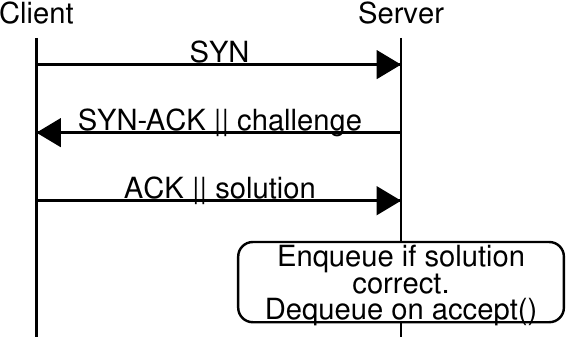}
    \label{fig:syn_puzzle}
  \vspace{-1em}
    \caption{}
    \vspace{1em}
  \end{subfigure}
  \caption{TCP three-way handshake with (a) no protection and (b) challenges enabled.}
\end{figure}

TCP SYN flood attacks aim to overflow a victim server's \texttt{listen} queue by
overwhelming it with half-open connection requests.  The attack forces the server to drop new incoming connections denying new clients
from receiving service~\cite{rfc4987}.
The attacker's sending rate should be high enough to overflow the server's
queue before the connection reset timers
expire.
A variant of the TCP SYN flood attack is a TCP connection flood in which an attacker
attempts to overflow the server's \texttt{accept} queue for the same purpose of denying legitimate clients
the opportunity to connect to the server. In a connection flood, the attacker completes the three-way handshake instead of leaving the connections half-open.


Among the server-based mitigations to SYN flood attacks, the SYN cache and TCP
SYN cookies are the most
common~\cite{rfc4987,lemon2002,cryptosyncookies}.
The SYN cache 
reduces the amount of memory needed to store state for
a half-open connection by delaying the
allocation of the full TCB state until the connection is established. Servers
implementing SYN caches instead maintain a
hash table for half-open connections that contains partial state
information and provides fast lookup and
insertion functions. 
SYN cookies, on the other hand, operate by eliminating the source of the vulnerability in TCP implementations: 
the state reserved for half-open connections in the TCB.
When SYN cookies are enabled, the server encodes a new TCP connection's parameters as a cookie in the packet's initial
sequence number, and refrains from allocating state for a new connection until the cookie is
again received from the client and validated. 

The SYN cache aims to contain TCP SYN attacks
by reducing the amount of state maintained on the server for half-open
connections. Although efficient against
a single attacker (or a small botnet), SYN caches do no provide protection
against larger botnets
for which the attack rate can easily exceed the space allocated for the cache.
Once the cache is full, the server
will default to the same behavior it performed when its backlog limit is
reached, defeating the purpose for which
the cache is used.
Although SYN cookies eliminate the key target of the SYN flood attack (the TCP
backlog), they do not provide protection against large botnets. Attackers in control of a
large number of zombie machines with valid (non-spoofed) IP addresses can, without added
effort, overload the server's \texttt{listen} queue with valid TCP requests at a rate that surpasses the server's ability to accept
them.
By only tackling the problem on the server end, 
SYN cookies do not present a mechanism to strip the malicious actors from the
ability to conduct exhaustion attacks; it is also not clear how SYN cookies can
be generalized to serve as protection schemes for different types of state exhaustion attacks~\cite{wang2003}.

\subsection{Client puzzles}

Cryptographic client puzzles have been proposed to counter an asymmetry in today's Internet: the clients can request
substantial server resources at relatively little costs.  Client puzzles alleviate this asymmetry by forcing clients to
commit compute power as payment
for the server resources that they are requesting, thus improving state exhaustion attacks resistance.


Client puzzles have been previously proposed as a mechanism to combat junk mail~\cite{dwork1993}, website metering~\cite{franklin1997}, protecting the network IP and TCP channels~\cite{juels1999,mcnevin2004,wang2003,feng2005},
protecting the TLS connection setup~\cite{dean2001,ietf-draft}, and protecting the capabilities-granting channel~\cite{Parno2007}.
Additionally, client puzzles are at the heart of the mining process of
nowadays' crypto-currencies~\cite{nakamoto2008,bonneau2015}.
Upon receiving a SYN packet, the server computes a puzzle
challenge and sends it back to the client. At this time, the server does not commit any resources to the initiating
client. After receiving the challenge, the client will employ its
computational resources to solve the challenge and
send the solution to the server. If the solution is correct, the server then commits resources to the client 
and accepts the connection.
Otherwise, the server
drops the connection.
For it not to break the three-way handshake, the server piggybacks the puzzle onto the SYN-ACK packet. The client can then solve the puzzle and send it along with its corresponding ACK packet.

Despite its promise, several challenges face the adoption of client puzzles as a practical measure of defense against
state exhaustion attacks. In the following, we discuss the challenges and our methods to address them.
First, there is a shortage of implementations that allow for the
comparison and the evaluation of different types of challenge creation and verification mechanisms.
In this paper, we implement 
clients puzzles in the Linux 4.13.0 kernel and provide access to our implementation in the form of a kernel patch. 

Second, an important advantage of client puzzles is the ability to increase the difficulty of the challenges as the intensity of the attacks increases. 
However, there are no concrete and theoretically-backed recommendations for selecting the appropriate challenge difficulties. 
The work in~\cite{wang2003} attempts to alleviate this problem by requiring clients to place bids (computing resources) on the server's resources. 
A client starts by asking for service without committing resources. If the server refuses to provide service, the client
solves increasingly harder puzzles until the server accepts its connection request.
The server may reply to a failed attempt with an acceptable puzzle hardness for the client to solve.
This approach suffers from two main drawbacks. First, it violates the TCP protocol by adding one or more extra round trips to the connection establishment 
phase. More importantly, this mechanism can be exploited to target clients as it moves the puzzle initiation process from the server to the client. 
An attacker would congest a client's egress links to trick her into believing that the server is requiring her to solve
harder puzzles. The attack will lead the client either to choose to refuse service or to commit more resources than
needed. In this work, we believe that the process of determining the puzzle difficulty should remain the server's
task. We, therefore, present a game-theoretic model that can be used to determine an appropriate difficulty given the server's provisioning and load.

Laurie and Clayton~\cite{laurie2004proof} present an economic analysis to argue against the use of proof-of-work mechanisms 
to combat email spam. We agree with the authors' conclusion that computational puzzles do not possess ``magical'' properties
making them practical in every situation, and that proof-of-work must be properly analyzed before adoption; which is what 
we sought to achieve in this work. We argue that memory- and computational-resource exhaustion attacks are of a different nature
than spam emails. First, unlike spam, attacker benefits from TCP state exhaustion attacks do not depend on the involvement of 
a human user to click on malicious links.
Second, DDoS attacks nowadays are mostly launched from compromised botnet machines and not
from specialized attacker hardware. Thus, beyond the initial compromise cost, launching an attack is virtually free for the attacker. 
We believe that our theoretical and experimental results showcase the merits of proof-of-work mechanisms in tolerating SYN and
connection floods.
In fact, our work complements the security analysis 
performed in~\cite{liqun2009security,groza2014crypto,douglas2011stronger} with the required protocol engineering 
and design steps, allowing for a rich and improved understanding of client puzzles for resilience to state exhaustion attacks. 


\section{The Game-Theoretic Model} \label{s:framework}
In this section, we introduce our game-theoretic model to compute the puzzle difficulty levels that balance the clients' computational load as well as the server's ability to combat TCP state exhaustion attacks and minimize its time to verify puzzle solutions. We first present the threat model and assumptions that we make in our research and then turn to discussing our game-theoretic model.

\subsection{Assumptions and threat model} \label{s:threat}

In this paper, we make the following assumptions.

{\em Assumption 1.}
Common state exhaustion attacks, specifically TCP connection floods as well as higher-layer attacks, require the
presence of a two-way communication channel between the attacker bots and the victim server. 
This is evident from the nature of the state exhaustion attacks as well as their ability to circumvent scrubbing
and filtering techniques by sending lower volumes of traffic~\cite{anreport2018}. This implies that during a
single-vector state exhaustion attack, the victim server is able to receive packets from and send packets to 
its legitimate users as well as the attackers' machines. In the presence of a multi-vectored attack, we assume
the presence of volumetric attack mitigation techniques (such as cloud-based protection-as-a-service); client
puzzles will complement those techniques to provide layered DDoS defenses to hybrid attacks. 


{\em Assumption 2.}
We assume that the attackers can control a large number of zombie machines that form botnets to
coordinate large-scale attacks aimed to deplete a target server's resources. 
However, we assume that that attacker's army of bots comprises commodity machines (e.g., workstations, mobile
phones, and IoT devices) but not servers or clusters with large computing resources. Such machines, being part of enterprise
solutions, are harder to compromise than commodity machines as they would employ better protective mechanisms.
We further assume that the attackers can capture and replay packets, but are not able to change their content. 
Protection against packet integrity attacks is beyond the scope of this paper.


The above assumptions are similar to the ones made in~\cite{juels1999} and~\cite{wang2003}.
Moreover, client puzzles do not require the end-server to
differentiate between the malicious and benign traffics. In fact, the low volume nature of state exhaustion attacks
and the requirement for quick and effective protective mechanisms can impede the accuracy of such detection mechanisms. 
Client puzzles, on the other hand, can provide timely protection against state exhaustion attacks as long 
as the benign clients are willing to
invest computing resources to receive service and thwart the attack.

\subsection{Difficulty Selection as a Stackelberg Game}
We formalize the problem of selecting the puzzle difficulty similarly to a network pricing problem~\cite{basar2002, shen2007_2, shen2007}. 
We model the problem as a Stackelberg game between the service provider and the service users. The service provider is the leader and is 
responsible for setting the difficulty of the puzzles that the clients must solve to receive service. The users are the followers who then 
choose their request rates to optimize their local utility. 

Our model rests on the assumption that all clients are selfish agents seeking to optimize their local utilities; we do
not specifically posit a model for malicious bots. This assumption is rooted in the following observations. 
First, before the attack starts, the server does not have means to distinguish between benign clients and malicious bots. 
Second, TCP by default treats every connection request it receives as a benign request, and thus sends an ACK packet back
without checking whether the request is sourced from a benign user or a compromised bot. Third, 
positing a specific attacker model would require the estimation of attacker preferences and utilities, to which the
server has no means of measuring. This could create a schism between the model and its application in the real world. We
therefore treat every request as if it is coming from a benign client, and capture the presence of a large botnet by
obtaining the asymptotic solution for our model. 


Let $x_i$ be user $i$'s request rate, for $i \in \set{1, 2, \ldots, N}$, where $N$ is the total number of users in the
system. Consequently, $x_{-i} = \sum_{j\neq i}^{N} x_j$ is the total request rate of all the other users. 
Our model captures the puzzle's difficulty using the expected number of hash operations needed to find and verify its
solution. Let $p_i$ be user $i$'s puzzle, $\ell(p_i)$ is then the expected number of hash operations that user
$i$ has to perform to find solution to $p_i$, and $S(\bar{x} = \sum_i x_i)$ be the expected service time for a user's
request.  User $i$'s utility can then be written as 
\begin{equation}
  u_i \parans{x_i, x_{-i}, p_i} = w_i \log(1+x_i) - \ell\parans{p_i} x_i - S(\bar{x})
  \label{eq:user_utility}
\end{equation}
$w_i$ is a user specific parameter that models the users valuation of the provider's service. In other words, $w_i$
represents the amount of work user $i$ is willing to pay per request. $\log \parans{1+x_i}$ represents the user's
expected benefit when making decisions under risk or uncertainty~\cite{sheng1984,basar2002}. The utility function can
be interpreted as the difference between the user's expected benefit and the amount of work she has to put to solve
a puzzle per request added to the expected service delay she incurs. 
Each user, being a rational and selfish agent, will choose a request rate that optimizes her local utility. This will
lead to the users adopting the Nash Equilibrium (henceforth referred to as equilibrium) rates $x^*_i$ for $i \in \set{1, 2,\ldots,N}$ such that 
\begin{equation}
  u_i \parans{x^*_i, x^*_{-i}, p_i} \geq u_i \parans{x_i, x^*_{-i}, p_i}, \; \forall x_i > 0, \forall i 
  \label{eq:user_ne}
\end{equation}

The service provider's problem is to find an appropriate puzzle difficulty such that (1) it can effectively reduce the
impact of state exhaustion attacks and (2) minimize the amount of work the server does to generate and verify puzzles. 
Let $\mathcal{P}$ be the space of all possible cryptographic puzzles and $g(p_i)$ and $d(p_i)$ be the expected numbers of hash
operations that the provider needs to perform to generate and verify a solution to puzzle $p_i$, respectively. 
We model the provider's problem as finding the
set of puzzles $\mathbf{p^*} = \set{p^*_i \in \mathcal{P}, \; i \in \set{1,2,\ldots,N}}$ such that 
\begin{equation}
  \mathbf{p^*} = \underset{\mathbf{p} \in \mathcal{P}^N}{\argmax} \sum_{i=1}^N \parans{\ell(p_i) - \parans{g(p_i) +
  d(p_i)}}x^*_i
  \label{eq:provider_utility}
\end{equation}
Equation~\eqref{eq:provider_utility} captures the provider's goal of maximizing the amount of work that the clients have
to perform to obtain service under attack while minimizing the amount of work it must perform to generate puzzles and verify solutions. 
This formulation, in fact, captures the trade-off between the puzzle's complexity and the expected work
that the provider needs to perform to generate and verify puzzles. 
The tuple $\parans{\mathbf{x^*}:=<x^*_1, x^*_2, \ldots, x^*_N>, \mathbf{p^*}:=<p^*_1, p^*_2, \ldots, p^*_N>}$ represents
the solution to the full Stackelberg game. 

We find the solution 
by first fixing $\mathbf{p}$ and finding the client's equilibrium request 
rates $\mathbf{x^*}(\mathbf{p})$. If such a solution exists, we can then solve for the optimal puzzles $\mathbf{p^*}$ by
using $\mathbf{x^*}(\mathbf{p})$ in Equation~\eqref{eq:provider_utility}. 

\section{Application to The Juels Puzzle scheme} \label{s:application}
We now show how the framework we introduced in Section~\ref{s:framework} can be applied to the puzzles protocol
presented in~\cite{juels1999}. We first describe the puzzles protocol 
from~\cite{juels1999} and then show the solutions we obtain using our framework. 
For our modeling and analysis, we assume that the server issues puzzles with the same 
difficulty for all of its clients, i.e., $\ell(p_i) = \ell(p_j) \;
\forall{i,j} \in \set{1,2,\ldots, N}$. This assumption ensures a stateless protocol and is following the IETF TLS 
puzzles draft~\cite{ietf-draft} and is recommended in previous research~\cite{juels1999}.

\begin{figure}[tb]
  \centering
    \includegraphics[width=0.85\columnwidth]{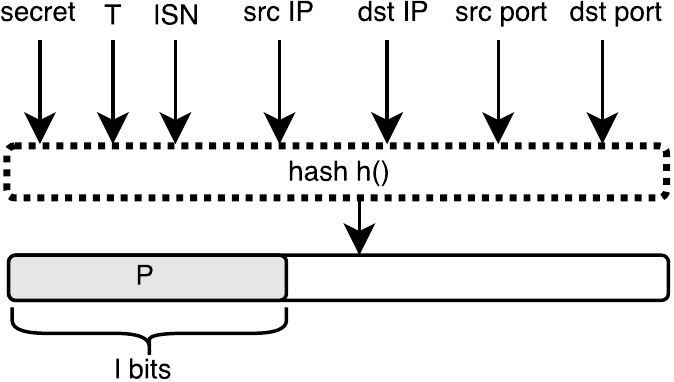}
  \caption{Puzzle construction~\cite{juels1999}}
  \label{fig:subpuzzle}
\end{figure}

A puzzle in this scheme is a bitstring of length $l$ bits having $m < l$ bits of difficulty. 
Figure~\ref{fig:subpuzzle} illustrates the construction of a challenge $P$.
The puzzle issuing server starts by
creating the hash $y = h\parans{s, T, \text{packet-level data}}$, where $s$ is a secret key; $T$ is a timestamp; packet-level data is a concatenation of the TCP Initial Sequence Number (ISN),
the source and destination IP addresses, and ports; and $h$ is a collision-resistant hash function. The server challenges a client to provide $k$ solutions to a puzzle $P$ formed by the first $l$ bits of $y$.


Upon receiving $P$, the client computes, by brute force, $k$ solutions 
$\set{s_1, \ldots, s_k}$ such that for $1 \leq i \leq k$, $|s_i| = l$ and the 
first $m$ bits of $h(P ~||~i~||~s_i)$ match the first $m$ bits of $P$, where 
$h$ is the same hash function that the server used -- $||$ denotes bit string concatenation. The client then sends
the solutions back to the server that in turn verifies their validity and 
subsequently accepts the request.


\subsection{The Solution}
Since obtaining a single solution of length $m$ bits is best done by brute force, it requires a maximum of
$2^m$ and an average of $2^{m-1}$ hashing operations.  Since each puzzle requires $k$ solutions, solving a puzzle then requires an average of $k\times 2^{m-1}$ hashing operations. 
Therefore, for each user $i \in \set{1,2,\ldots, N}$, $\ell(p_i) = k\times 2^{m-1}$.

To capture the expected service time for the users, called $S(\bar{x})$, we abstract the server's operation by an $M/M/1$ queue with a service rate $\mu$. We argue that this abstraction is enough for our purpose since the attacks we are interested in target the TCP stack and are independent of the application that the server is running; they are only affected by the application's 
ability to remove established connections from the accept queue. The service rate $\mu$ can be obtained by running stress tests on the application provider's infrastructure and can capture different service optimizations such as replications and caching. 
Subsequently, we express the expected service time as $S(\bar{x}) = \frac{1}{\mu - \bar{x}}$, when $\bar{x} < \mu$.
This condition assumes that the server is well-provisioned to handle the users' load under regular conditions.

Therefore we rewrite Equation~\eqref{eq:user_utility} as 
\begin{equation}
  u_i \parans{x_i, x_{-i}, p_i} = w_i \log(1+x_i) - k\times 2^{m-1} x_i - \frac{1}{\mu - \bar{x}}
  \label{eq:app_user_utility}
\end{equation}

We now turn to the provider's formulation. We represent the space of all possible puzzles as the set of tuples $(k,m)$
where $k \in \mathbb{N}$ is the number of solutions requested and $m \in \mathbb{N}$ is the number of bits of
difficulty in each. 
Therefore we write $\mathcal{P} = \set{\parans{k,m}, k,m \in \mathbb{N}}$.
As previously discussed, every challenge can be generated using only one 
hash operation, therefore we write $g(p_i) = 1, ~ \forall{i}$.  

When the server receives a solution, it generates a hash from the received packet's header and then verifies each of the
$k$ solutions until it finds a violating one or deems the puzzle correctly solved. If the server chooses which of the
$k$ solutions to verify uniformly at random, it then needs an average of 
$\frac{k}{2}$ hashing operations. Therefore, we can write $d(p_i) = 1 + \frac{k}{2},~\forall{i}$. 

Since we assume that the service provider issues puzzles with the same difficulty for all users, we henceforth write $p
= (k,m) = p_i, \forall{i}$. We can then rewrite Equation~\eqref{eq:provider_utility} as 
\begin{equation}
  p^* = \underset{p \in \mathcal{P}}{\argmax} \sum_{i=1}^N \parans{k\times 2^{m-1} - 2 -
  \frac{k}{2}}x^*_i(p)
  \label{eq:app_provider_utility}
\end{equation}
Let $w_{av}$ be the average client valuation of the server's service and $\alpha$ be the server's asymptotic service rate per user, under normal operation.
\begin{thm} \label{thm:nash}
  The Nash equilibrium is achieved at $p^* = (k^*, m^*)$ such that:
    \begin{equation} \label{eq:theorem_diff}
      \ell(p^*) = k^*\times 2^{m^*-1} = \frac{w_{av}}{(\alpha + 1)}
    \end{equation}
\end{thm}
\begin{proof}
We show the proof of Theorem~\ref{thm:nash} in Appendix~\ref{app:proof}.
\end{proof}

\subsection{Analysis}
The equilibrium difficulty we obtained in Theorem~\ref{thm:nash} illustrates an important design tradeoff between the server's
provisioning and the difficulty of the puzzles that the clients should solve when the server is under attack. A well-provisioned server, 
i.e., one for which $\alpha > 1$, will be able to absorb a larger fraction of the attack and subsequently asks its clients to solve less complex challenges. In that case, the clients help the server tolerate the attack and commit fewer resources than they are willing to (the average number of hashes they would need
to perform to solve a challenge is less than $w_{av}$) --- the client achieves high utility.
On the contrary, a server that is not able to handle all of its clients' regular load, i.e., one for which 
$\alpha < 1$, would require its clients to solve harder puzzles ($p^* \simeq w_{av}$)  and thus achieve lower utility levels. Therefore, to tolerate an attack, the server asks its clients to commit more resources risking more clients dropping out as the intensity of the attack increases. Those clients with $w_i < w_{av}$ would consider it more beneficial for them to drop out since it would be too costly as a function of the resources committed to obtaining a connection.

We further note that our model and solution are agnostic to the application that is run by the server as well as the specific server configuration.
This, in fact, is consistent with TCP being a transport layer protocol that is independent of the type 
of application running on top of it. 
All our model requires is an estimate of the server's capacity to handle large loads (i.e., the parameter $\alpha$) which can be obtained by running appropriate
stress tests. Server replication and load balancing are then captured in our model through an increase in the value of $\alpha$ (given the same load).

Finally, we note that our result is not affected by the presence of long-lived TCP connections (for example, if using HTTP/1.1~\cite{rfc7230http}). The puzzles protect the TCP connection 
establishment channel and allows users to connect to the server in the presence of malicious attacks. The lifetime of the established connection is not affected by the presence of puzzles or lack thereof; 
in the case of HTTP/1.1, the goal of the challenges is to allow clients to establish the TCP connection upon which the HTTP session persists. 
Moreover, the solution we present in Theorem~\ref{thm:nash} captures $p^*$ in terms of expected number of hash operations that a client needs to perform {\em per attempted connection}, while 
$w_{av}$ represents the average client valuation of the requested service {\em per request}.  For an HTTP/1.1 persistent session, the client would only need to pay $p^*$ hashes once. 

\subsection{Obtaining model parameters}
\begin{figure}
  \centering
  \resizebox{1.0\columnwidth}{!}{%
    \includegraphics{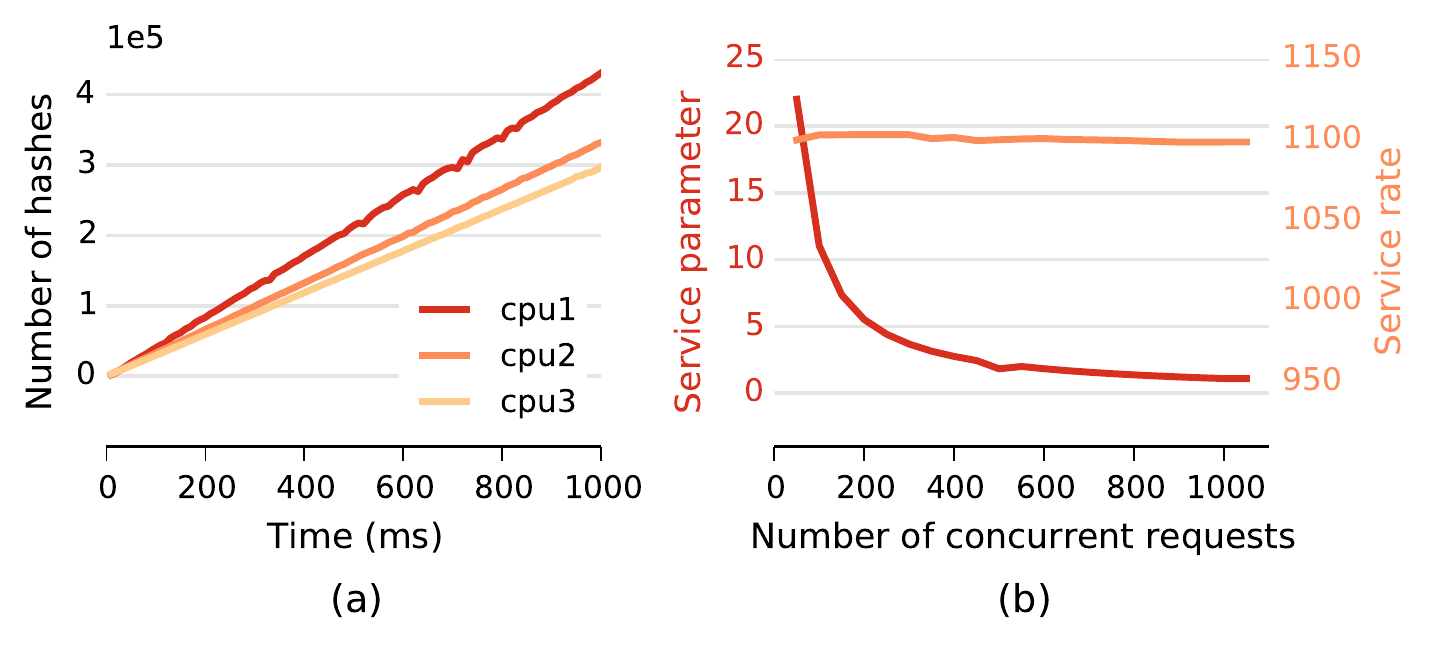}
  }
  \caption{Performance profiles of (a) client ($w_{av}$) and (b) server ($\alpha$).}
  \label{fig:nash_exper}
\end{figure}
The model parameters, $w_{av}$ and $\alpha$, relate to the performance capabilities of the server and the clients. We provide an experimental procedure to obtain the model parameters. Then we apply the procedure to an experimental setup to show the Nash strategy.

First, $w_{av}$ is the number of hashes we assume the client is willing to perform to complete the TCP handshake. It represents the level of acceptable service degradation as each TCP connection will take longer to finish. To find $w_{av}$, we assume that 400 ms is adequate time to establish a TCP three-way handshake for a legitimate client when the server is under attack. Usability studies show that a 400 ms delay does not interrupt the user's flow of thoughts~\cite{book:usability}. Using this assumption, we find the number of hashes a machine can perform for 400 ms by profiling the machines. $w_{av}$ is the average value obtained during the experiments.

Second, $\alpha$ is the service parameter of the server. It is directly related to the processing power of the server. To obtain the parameter, we start by stress testing a server. The stress test varies the rate of requests per second and records the time it takes to get service for each rate. We compute $\alpha$ as the ratio of service rate over the number of concurrent requests.

Finally, after obtaining $w_{av}$ and $\alpha$, we calculate the equilibrium difficulty parameters $(k^*,m^*)$ using Equation~\eqref{eq:theorem_diff}. 
The choice of those parameters exposes a tradeoff between the number of hashes the server needs to verify a solution and the probability that an adversary can guess a solution. Choosing a very small $k$ will increase the attacker's ability to guess a solution, and selecting a large $k$ will increase the solution verification time. On the other hand, by selecting lower values of $k$, the challenge difficulty $m$ would increase allowing the server to offset its lack of computational resources onto its clients by asking them to solve harder challenges.

\subsection{Example}
In the following, we present an example for computing the Nash equilibrium difficulty for a server serving a variety of machines with varying processing powers. Starting with the client, we obtain $w_{av}$ by profiling the number of SHA-256 operations per second.  Figure~\ref{fig:nash_exper} (a) shows the profile of three CPU types:
(1) \texttt{cpu1} is an Intel Xeon E3-1260L quad-core processor running at 2.4 Ghz, 
(2) \texttt{cpu2} is an Intel Xeon X3210 quad-core processor running at 2.13 Ghz, and
(3) \texttt{cpu3} is an Intel Xeon processor running at 3Ghz. 

The average number of hashes that can be performed over the three types of CPUs is $w_{av} = 140630$. Although the CPUs we profiled are not an exhaustive representative set of the processing powers of a typical clientele,  in all of our experiments, we leveled the play-field by providing all the attackers with similar or better computational powers.  In other words, if a client can perform a three-way handshake in 400 ms,
a typical attacker can perform the same connection in a time equal to or less than 400 ms.  

Then we estimate the server's $\alpha$ parameter. We deployed an Apache2 web server on a dual Intel Xeon hexa-core processor running at 2.2 GHz with 24 GB of RAM.  We then used the apache benchmarking tool \texttt{ab}~\cite{abtool} to profile the performance of the server under regular and high loads. Figure~\ref{fig:nash_exper} (b) shows the service rate ($\mu$) and the service parameter $\alpha$ of our server as the number of  concurrent requests attempted by \texttt{ab} increases. Our server was able to maintain a constant service rate under high load ($\mu \simeq 1100$ requests/s), and thus the parameter $\alpha$ converged to a value of $1.1$ as the load increased.
%
%
%
Thus for our example, with $w_{av}=140630$ and $\alpha=1.1$, the TCP puzzle difficulty is set at the value $k^* = 2$ and the difficulty $m^* = 17$ bits from Equation~\eqref{eq:theorem_diff}. That is, each challenge requests two solutions with 17 bits of difficulty each.

\section{Implementation} \label{s:implementation}
We implemented the TCP challenges in the Linux kernel's TCP stack using the Linux 4.13.0 source. The puzzles are turned off by default and are only enabled when the socket's queue is full. We designed our implementation in a way such that the challenges take precedence over the SYN cookies once the queue is full; we do however support SYN cookies as a backup option.  We provided support for dynamically tuning the parameters of the challenges through the kernel's \texttt{sysctl} interface. Both $k$ and
$m$ can be adapted at any point during the server's runtime. 

We generate the challenge's pre-image by hashing a string containing: (1) a server's secret key, generated once at the start of a socket's lifetime, (2) the server's current timestamp, (3) the SYN packet's source and destination IP addresses, and (4) the SYN packet's source and destination port numbers. 
For the hashing function, we used the Linux kernel's SHA256 implementation since it provides the necessary pre-image resistance guarantees~\cite{juels1999}.

\begin{figure}
    \centering
     \begin{bytefield}[bitwidth=0.80em]{32} 
      \bitheader{0-31} \\
\bitbox{8}{Opcode 0xfc} & \bitbox{8}{Length} & \bitbox{8}{$k$} &
\bitbox{8}{$m$} \\
        \bitbox{8}{$\ell$} & \bitbox{24}{Preimage $\cdots$} \\
        \bitbox{8}{$\cdots$ preimage} & \bitbox{24}{Padding (NOP)} 
     \end{bytefield}
     \caption{TCP Options block for a SYN challenge.}
     \label{fig:tcpoptchallenge}
\end{figure}

\begin{figure}
    \centering
     \begin{bytefield}[bitwidth=0.80em]{32} 
      \bitheader{0-31} \\
\bitbox{8}{Opcode 0xfd} & \bitbox{8}{Length} & \bitbox{16}{MSS value} \\
        \bitbox{8}{Wscale} & \bitbox{24}{Solution $\cdots$} \\
        \bitbox{8}{$\cdots$ solution} & \bitbox{24}{Solution $\cdots$} \\
        \bitbox{8}{$\cdots$ solution} & \bitbox{24}{Padding (NOP)}
     \end{bytefield}
     \caption{TCP Options block for a SYN solution with $k=2$.}
     \label{fig:tcpoptsolution}
\end{figure}

In order not to break the TCP definition, we inject the challenges and solution into the options field of the TCP SYN-ACK and ACK packets. Figure~\ref{fig:tcpoptchallenge} shows the format of the TCP option we implemented to transmit a challenge in the SYN-ACK packet. We chose an unused opcode (0xfc) to represent a challenge option. The \texttt{Length} field indicates the length of each option block in bytes, including the opcode and the field itself. We allocate one byte each for the number of solutions $k$, the difficulty of the puzzle $m$ (in bits), and the pre-image and solution length $l$. Next, we insert the challenge's pre-image. Finally, following the TCP stack requirement, each option block must be 32 bits aligned, we, therefore, insert 0 to 3 NOP fields to ensure alignment. 

Figure~\ref{fig:tcpoptsolution} shows the format of the TCP option used by a client to send a computed solution. Similar to the challenge option, we made use of the unallocated opcode (0xfd). Since the server keeps no state about the client after receiving the first SYN packet, the client safely assumes that the server has ignored its previously announced {\em Maximum Segment Size} (MSS)
and {\em Window Scaling} (Wscale) parameters. We then resend the MSS and Wscale values within the solution block and then write down each of the $k$ solutions and perform alignment to 32 bits. 

The benefits of adding the MSS and Wscale parameters to the solution option block are twofold. First, the challenge protocol would be self-contained; implementation of the TCP stack usually ignores all options other than timestamps in any packet other
than the SYN and SYN-ACK packets. Therefore supporting the challenge protocol does not require changes to legacy options parsing. The addition also provides us with the benefit of reducing the space needed to resend the options in the ACK packet. For example, sending the MSS values as a separate option would require
4 bytes while we only need 2 in the case of the self-contained solution option. Second, we encode the MSS value using 16 bits (as defined in the specification of TCP), instead of the 3 bits provided by SYN cookies. Additionally, when SYN cookies are in
operation, the client and the server can no longer agree on the window scaling parameters, which reduces the performance of the TCP connection.

Also, modern implementations of the TCP stack support the exchange of timestamps as options in the TCP header to improve the estimation of the connection's round-trip time. Our implementation makes use of the timestamps option, whenever available, to generate, solve, and verify a challenge.
However, in the case where the timestamps option is not put to use
(for example turned off by the client or the server), our implementation embeds the timestamp used in the generation of the challenge (an additional 4 bytes) in both the challenge and the solution option blocks.


Furthermore, when the server's accept queue overflows, its default behavior is to reject new connections, even if the protection mechanism is in place. However, for our purposes, since the goal of the puzzle protection mechanism is to throttle the rate of all clients (both benign and malicious), we modified the listening TCP socket's implementation to send a challenge when the protection is in effect, even if the accept queue overflows. When the server receives an ACK packet while under attack, it first checks if the queue is full and only performs the verification procedure when there is room to accept the connection. If the queue is full, the server will ignore the ACK packet. In such a case, the user (both benign and malicious) assumes that the connection has been established and will
begin sending application level packets thus causing the server to reply with a reset (RST) packet to signal that the connection was not established. 
This implementation choice achieves the goal of deceiving
the malicious users that they have established a connection while they have not; the malicious agents that do not send application-level packets will not receive a RST packet to indicate that the server has dropped the connection. 

Finally, to combat replay attacks, we make use of the timestamp in the solution to check if a challenge has expired. This stateless mechanism hinders an attacker's ability to replay solution packets since tampering with the timestamp will cause the solution verification to fail. 
The timeout interval can be tuned through the kernel's \texttt{sysctl} interface. 



\section{Evaluation} \label{s:evaluation}
Using our modified Linux kernel, we evaluated the performance of the 
TCP puzzles in safeguarding a server TCP connection establishment channel
against both SYN and connection floods. The difficulty of the puzzles 
employed is the Nash solution that we established in Section~\ref{s:application}. 

 
We performed the experiments using DETER~\cite{deter}, a cybersecurity testbed based on Emulab provided by the University of Southern California. DETER allows for reproducible experiments that are described using {\em network simulator} (ns) scripts and an agent activation module. The testbed automatically deploys and executes the scenarios. 

The goals of our experiments are to evaluate (1) the effectiveness of TCP puzzles in the protection against state
exhaustion attacks, (2) the impact of TCP puzzles on service quality, and (3) the ability of the Nash equilibrium puzzle difficulty to balance the client solution and the server verification loads as well as the
ability to effectively rate limit malicious attackers. 
We perform the experiments using a test deployment of a single server providing service to a set of clients while being the target of a state exhaustion attack launched by a botnet of malicious machines. 
Specifically, our server runs an apache2 HTTP server with an application that accepts ``\texttt{gettext/size}" requests and returns messages containing \texttt{size} bytes of random text.
The clients, on the other hand, run an HTTP client that requests text from the server at a pre-specified rate. 
 
While in a real-world deployment, service would be provided by a farm of servers, our scenario uses only one server and a smaller set of clients. In those larger systems, since a load balancer forwards TCP connection requests to the individual servers, an attack has to ensure its wave of requests reaches all of the servers to effectively deny service. Therefore, adding more servers allows a service provider to tolerate bigger attacks by large botnets. Our results show that, in essence, a server using TCP puzzles as means for state exhaustion DDoS protection can tolerate a larger botnet than an unprotected server. 
We hence argue that when all the servers in a farm employ our protection, the system will be able to tolerate a larger botnet that is proportional to the improved tolerance of a single server.  Our experiment scenario thus studies the protection offered to a single server; the results are to scale when more puzzles-equipped servers are deployed in a load balancing scheme.
 

%
We consider two types of attackers, the first uses randomized source IP addresses to target the server's \texttt{listen} queue with a flood of half-open TCP connections (using \texttt{hping3}). The second type uses real IP addresses to flood the server with established connections (using \texttt{nping}) in an attempt to fill its \texttt{accept} queue and prevent new, legitimate, connections from being established.
Unless otherwise stated, we use the following experiment parameters. The set of clients contains 15 machines requesting
$10,000$ bytes of data at exponentially distributed time intervals, with rate $r_c=20$ requests per second. The botnet
consists of 10 machines running an attack at a constant rate $r_a=500$ requests per second, amounting to an overall attack rate of $5,000$ packets per second (pps). All of the malicious machines are equipped with a computational power equal to, or greater than, that of the clients' machines. All of the machines in our setup are equipped with our modified kernel, with the exception of Experiment 5 in which we study the impact of the puzzles when some of the machines do no deploy the patch. 


Finally, with the exclusion of Experiment 6, all the experiments use the same network topology with well-provisioned link bandwidths so as to avoid link saturation. The backbone consists of three routers fully connected with 1 Gbps links. The server connects to the network with a 1 Gbps link while all the other hosts connect to the network with 100 Mbps links. 
All of our agents run on physical machines with Ubuntu 16.04 LTS along with our patched Linux 4.13.0 kernel. We provide more details about the hardware specs of the machines we used as well as the network topology (Figure~\ref{fig:topology}) in Appendix~\ref{app:topology}. 
%
%
We deploy the packet monitoring software, \texttt{tcpdump}, on all of the machines and use the captures to measure the throughput at the server, the throughput at each host, the TCP connection time, and the number of dropped TCP connections. We elect to report on the throughput since it represents a direct assessment of the impact of puzzles on our application, nevertheless we acknowledge that different applications will require different metrics. 
 
\subsection{Experiment 1: Impact of puzzles on client performance}
In the first experiment we show that the puzzles' impact on the connection time can be controlled by setting the
parameters $k$ and $m$. We study the impact of TCP puzzles by varying the number of solution required per challenge, $k$, 
over the set $\{1,2,3,4\}$ and the number of difficulty bits per solution, $m$, over the set $\{4,10,16,20\}$.

\begin{figure}[t]
\center
	\includegraphics[width=0.9\columnwidth]{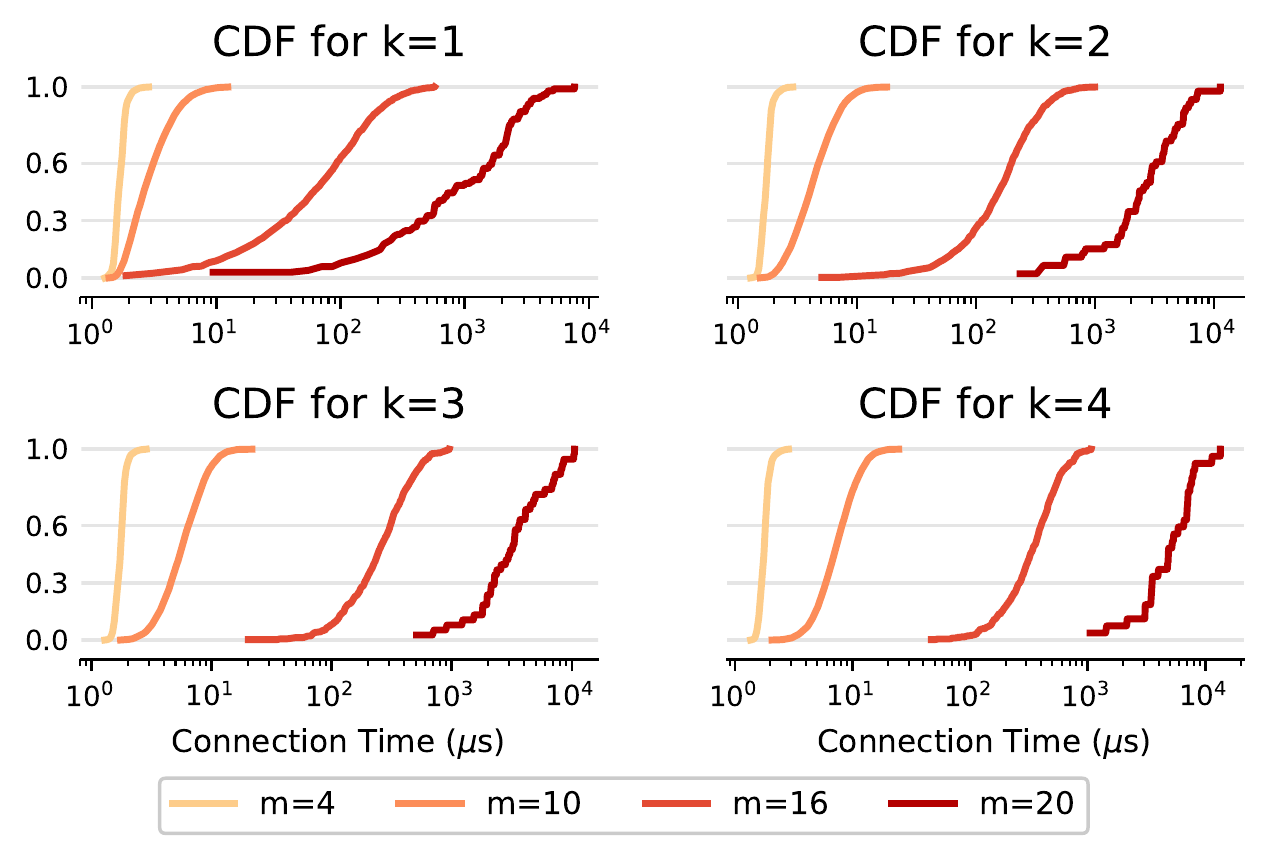}
	\caption{CDF of the connection time}
	\label{fig:exp1:cdf}
\end{figure} 

Figure~\ref{fig:exp1:cdf} shows the cumulative density function (CDF) of the connection time of a client node as the parameters $k$ and $m$ are varied. 
We first note that any increase in either parameter causes the connection times to increase. However, each parameter
affects the rate of change of connection time with a different magnitude. More to the point, increasing the number of
difficulty bits increases the connection time with an exponential factor. For example, when $k=1$, the average
connection time for $m=4$ is 2.0 $\mu$s while it is 286 $\mu$s for $m=16$. On the other hand, changing the number of
puzzles increases the connection time with a constant factor, in $m=16$ the average connection time for $k=1$ is 286
$\mu$s while the average connection time for $k=4$ is 558 $\mu$s. By tuning both variables of the puzzle difficulty, the
defender has a fine-grained control over the connection time for a host and thus its ability to perform state exhaustion
attacks.

\subsection{Experiment 2: SYN and connection flood protection} \label{s:eval:main}
In the second set of experiments, we show that a server running TCP puzzles is able to tolerate a SYN flood and a connection flood. We also show that TCP cookies do not offer protection during a connection flood.

\begin{figure}[t]
\center
	\includegraphics[width=1\columnwidth]{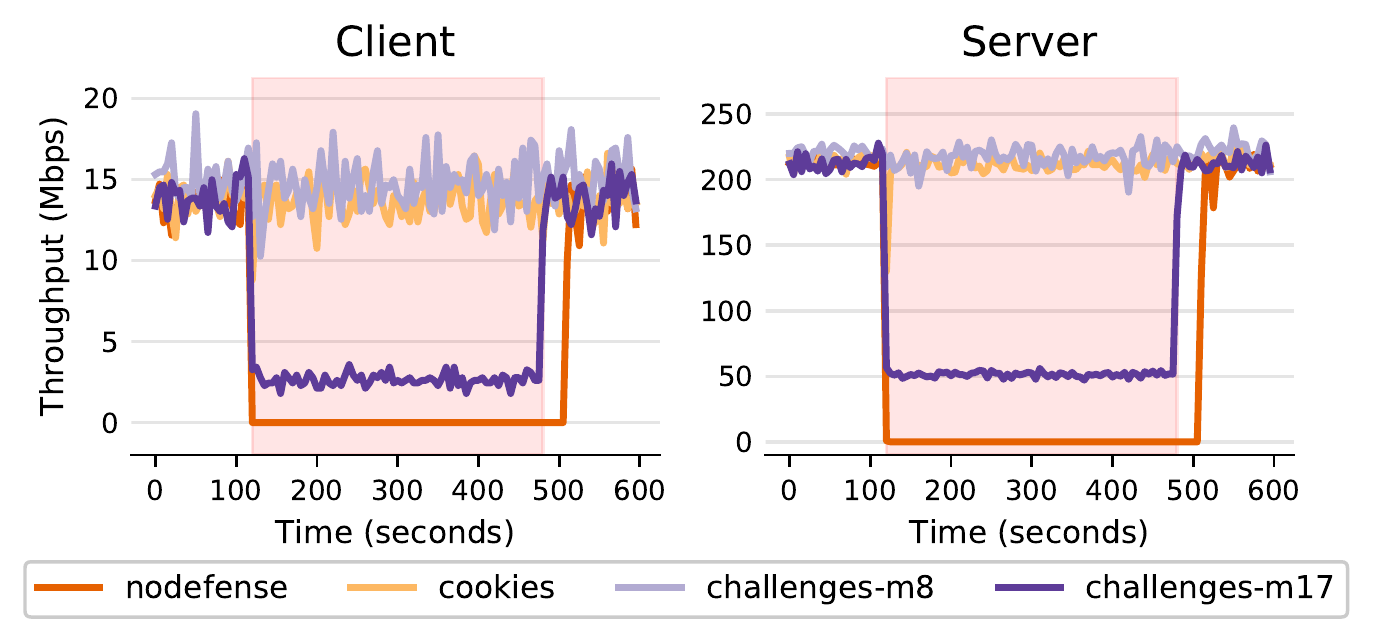}
	\caption{Throughput at a client and server during SYN flood.}
	\label{fig:exp2:throughput}
\end{figure} 

\begin{figure}[t]
\center
	\includegraphics[width=1\columnwidth]{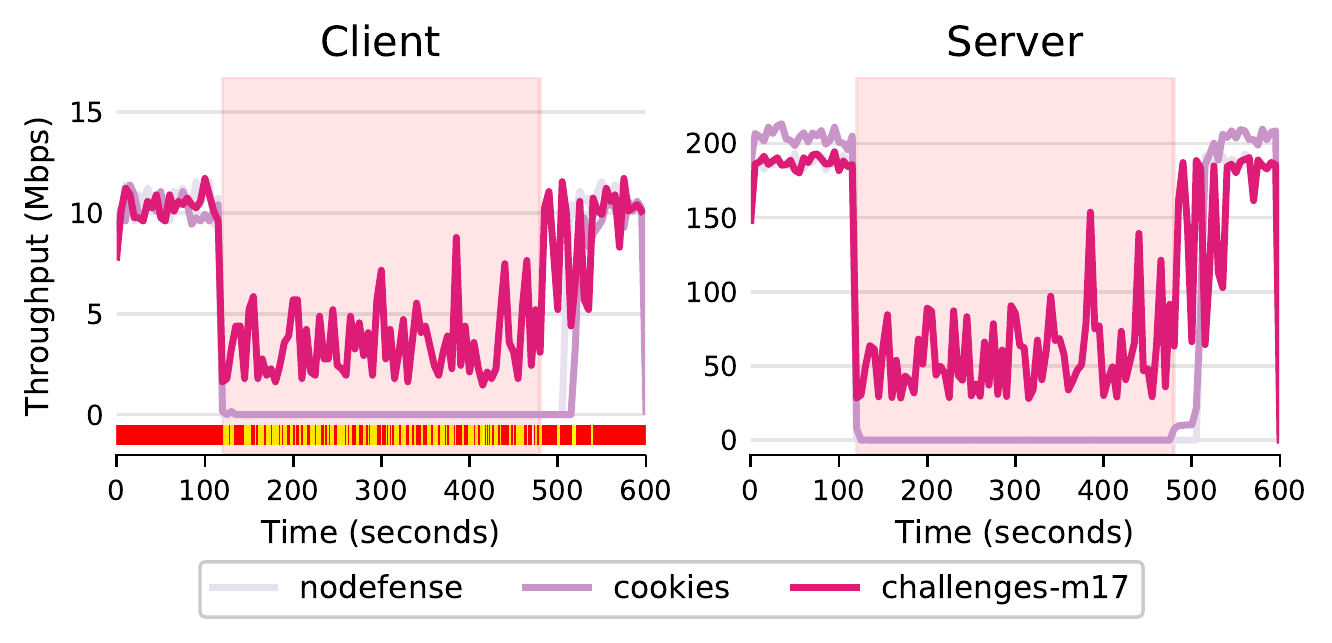}
	\caption{Throughput at a client and server during a connection flood.}
	\label{fig:exp2:2:throughput}
\end{figure} 

\begin{figure}[t]
  \centering
  \includegraphics[width=1\columnwidth]{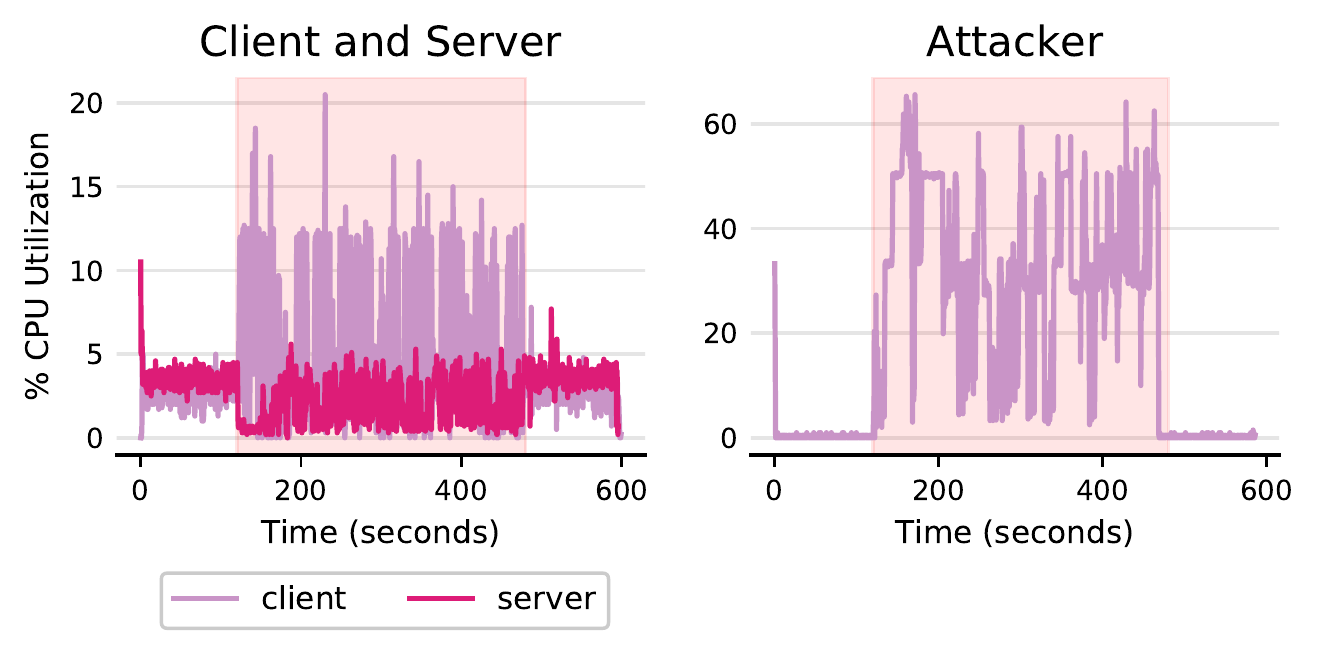}
  \caption{Impact of client puzzles on the cpu utilization during a connection flood attack.}
  \label{fig:exp2:cpu}
\end{figure}

\begin{figure}[t]
  \centering
  \includegraphics[width=1\columnwidth]{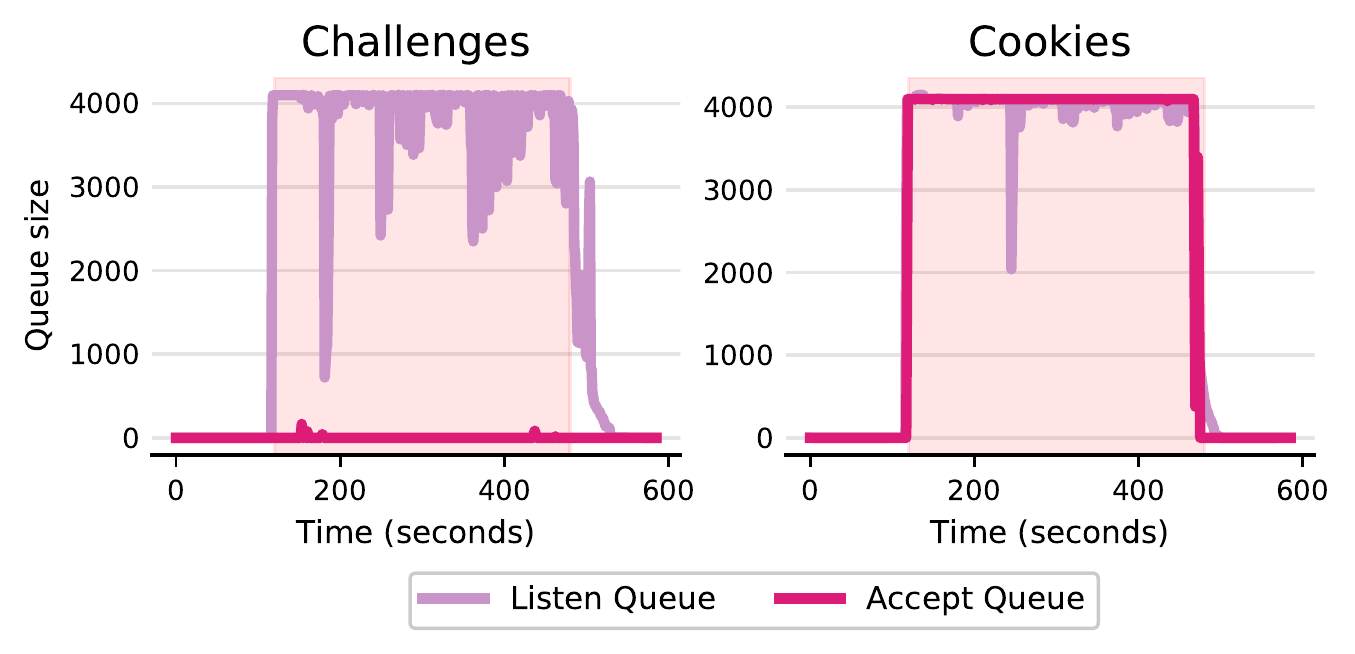}
  \caption{Listen and Accept queue size during a connection flood attack.}
  \label{fig:exp2:queue}
\end{figure}

\begin{figure}[t]
    \centering
    \includegraphics[width=1\columnwidth]{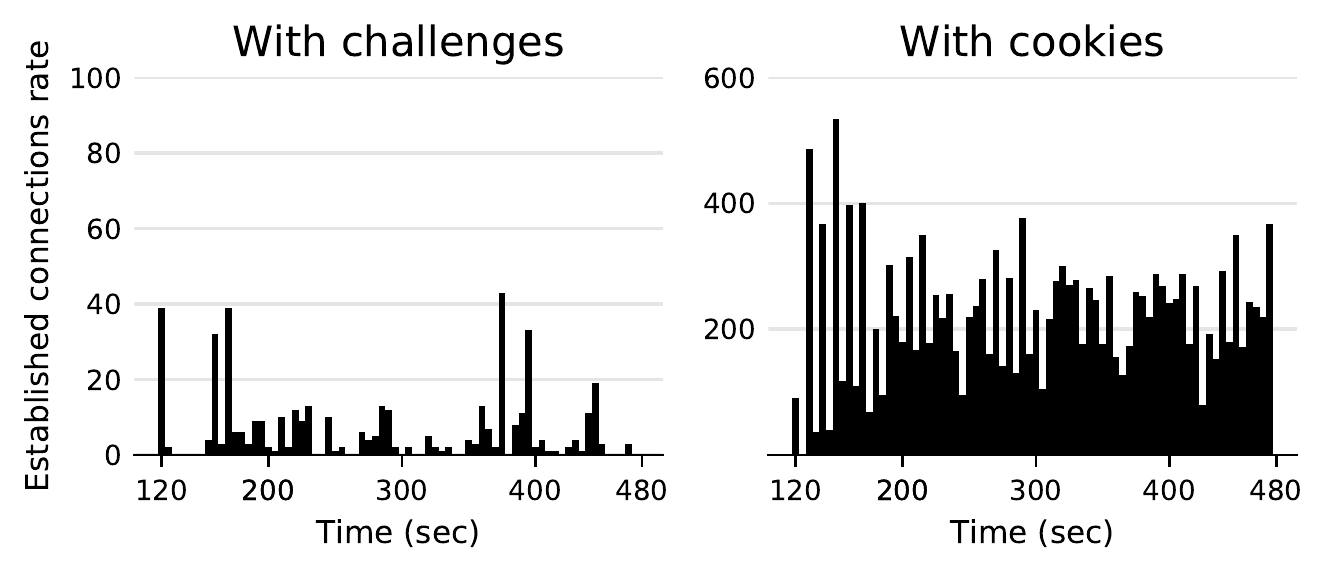}
    \caption{Effective attack rate for all attackers during a connection flood attack.}
    \label{fig:exp2:attacker}
\end{figure}

In the first scenario, we start a distributed SYN flood attack. Without protection, the SYN flood fills the
\texttt{listen} queue
with half-open TCP connections leading the server to drop new incoming connections. We measure the throughput at a
client and the server for three settings: (1) no protection (control settings), (2) TCP SYN cookies, and (3) TCP client
puzzles. Figure~\ref{fig:exp2:throughput} shows the throughput measured during the experiment. The attack duration,
shown by the shaded region, is initiated at $t=120$ and concludes at $t=480$. The throughput's behavior for both the
server and client are consistent; we therefore restrict our analysis to the server's case. For the control setting, the
server's throughput drops to zero as soon as the attack starts and returns to full capacity $30$ seconds after the
attack ends. On the other hand, SYN cookies are effective at rendering such an attack 
ineffective and ensuring a constant throughput at the server throughput the attack. By storing partial state of the
connection in the TCP sequence number instead of in the \texttt{listen} queue, SYN cookies provide protection against this type
of attack. Finally, when low difficulty puzzles are enabled, $(k,m)=(1,8)$, the throughput is unaffected during the
attack. Similarly to SYN cookies, the puzzles enable to reconstruction of a connection's state with no use of the
\texttt{listen} queue. However, when using the Nash equilibrium difficulty, $(k,m)=(2,17)$, the throughput is reduced to 50 Mbps during
the attack. The throughput reduction is due to the Nash equilibrium strategy being more aggressive than the easier
setting; in this scenario easy puzzles were enough to alleviate the attack as the botnet is not completing the
connection, this would not be the case for the next scenario.

In the second scenario, we use the attacker nodes to launch a distributed connection flood attack. We measure the same
metrics as the first scenario for three cases: no protection, SYN cookies, and TCP puzzles at Nash difficulties. The TCP
puzzles at difficulty of 8 bits were ineffective at protecting the server's state. For readability, we elected not to
show these results in this plot since we will revisit various difficulty settings in Section~\ref{sec:nash_diff_st}.

Figure~\ref{fig:exp2:2:throughput} shows the throughput of a client and the server during the experiment. Moreover, we
use the sparkline in the client plot to mark when the server sends a SYN-ACK packet with a challenge (bright tick) or
without a challenge (dark tick). The results show that SYN cookies are ineffective during a connection flood, the
server's throughput drops to $0$ as is the case when no protection is in place. 
In both those cases, the server needs $30$ seconds to detect the end of the flood and fully recover. On the other hand,
TCP puzzles at Nash difficulties provide tolerance against the flood attack. The throughput of the client and the server
is about 40\% of their respective nominal rates. It is interesting to note that the throughput periodically spikes
during the attack phase. This occurs because not all the requests of the clients require a puzzle, as shown by the dark
ticks in the sparkline during the attack phase. The performance improvement is due to the opportunistic nature of the
protection controller; that is, when the \texttt{listen} queue is not full, a connection request is answered without a challenge
allowing a host to take advantage of the resource instantly. We also note that easy puzzles were unable to affect the
attacker bots' connectivity rates and thus provided no better protection than SYN cookies. 

Additionally, we measured the impact of the TCP challenges on the CPU utilization of the client, server, and attacker
machines. Figure~\ref{fig:exp2:cpu} shows that the impact on the server of generating and verifying the puzzles is
negligible, the server's CPU utilization stayed below 5\% and did not exceed its nominal (under regular load) value. 
In accordance with the nature of computational puzzles, the CPU utilization at the clients' machines increased during 
the attack, nevertheless still remaining well under 20\%, with an average of 10\%. The attacker machines, on the other
hand, witnessed a spike in CPU utilization during the period of the attack, reaching a maximum of 60\%.
These results highlight that our equilibrium difficulty setting achieved our desired goals of (1) putting minimal
overhead on the server to generate and verify puzzles, (2) inducing tolerable nuisance to the clients, and (3)
effectively rate limit the attackers' established requests rate and increase their computational burden. In fact, the
sudden increase in the CPU utilization at the botnet machines can serve as an alert to the owners of these machines to the
possible presence of malware.

We further study the impact of the TCP cookies and puzzles on the server's internal \texttt{listen} and \texttt{accept}
queues during a connection flood attack. Figure~\ref{fig:exp2:queue} illustrates that when SYN cookies are the only defensive
mechanism in place, both queues are fully saturated which explains the zero throughput observed by the benign clients. 
On the other hand, with TCP challenges in place, the \texttt{accept} queue is almost always empty, which is a direct result of
the puzzles being able to rate limit every user, whether benign or malicious, to an average of 2 requests per second. 
Additionally, the \texttt{listen} queue, although mostly saturated, shows frequent openings that are consistent with the
opportunistic nature of our implementation as highlighted by the sparklines in Figure~\ref{fig:exp2:throughput}.

Finally, we show that TCP puzzles (at Nash difficulty) throttle the attacker's rate of established connections. We measured the effective
completed connections rate of all attackers as seen by the server during the connection flood. The measurements, shown
in Figure~\ref{fig:exp2:attacker}, reveal that the attack rate is not affected by TCP cookies, achieving an average rate
of 225 connections per second (cps), whereas TCP puzzles severely limit the attackers' rate at an average of 4 cps, a
reduction by a factor of 37.

\subsection{Experiment 3: Nash equilibrium strategy} \label{sec:nash_diff_st}

\begin{figure*}[t]
  \centering
  \includegraphics[width=\textwidth]{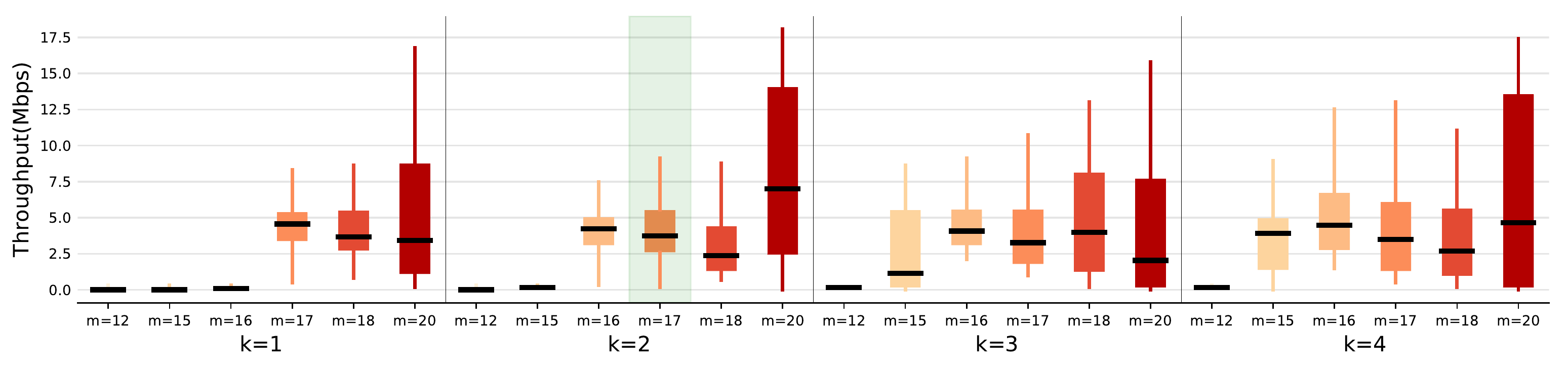}
  \vspace{-1em}
  \caption{Box plot of the client throughput for different puzzle difficulties.}
  \label{fig:nashbox}
\end{figure*}

In this
experiment, we show that the Nash equilibrium difficulty provides the optimal balance between the clients' throughput and the attack tolerance during an attack; the Nash equilibrium
is selected based on the capabilities of the clients and the defense requirements of the server.
We study the impact on the clients' throughput and the attacker's abilities of the Nash equilibrium strategy compared to other difficulties during a connection flood attack.

Figure~\ref{fig:nashbox} shows the average and standard deviation of the throughput of a client during attack. In general, for any k, if $m<12$, the ease of solving the challenges does not affect the attackers' rate, thus causing a denial of service. The Nash equilibrium strategy results in the most stable throughput with an average of $3.90$ Mbps and low variability. 
Even though some of the other settings have a higher average throughput, the throughput is highly unstable reaching zero at many times. Additionally, 
we note that when the difficulty is set to $(k=2,m=16)$, the throughput achieves a slightly better average with
comparable variability. However, the Nash difficulty setting provides the rate that balances the acceptable cost a
client is willing to pay (in terms of increased connection time and thus decreased throughput), and the server's ability
to tolerate state exhaustion attacks by throttling the attackers' rates. In fact, at the Nash difficulty, the puzzles mechanism reduced the attackers' average SYN sending rate from 2250 pps for $(k=2,m=16)$ to 1668 pps, and the average connection establishment rate from $30$ cps to $22$ cps. 

\subsection{Experiment 4: Botnet effectiveness}

\begin{figure}
\centering
  \includegraphics[width=\columnwidth]{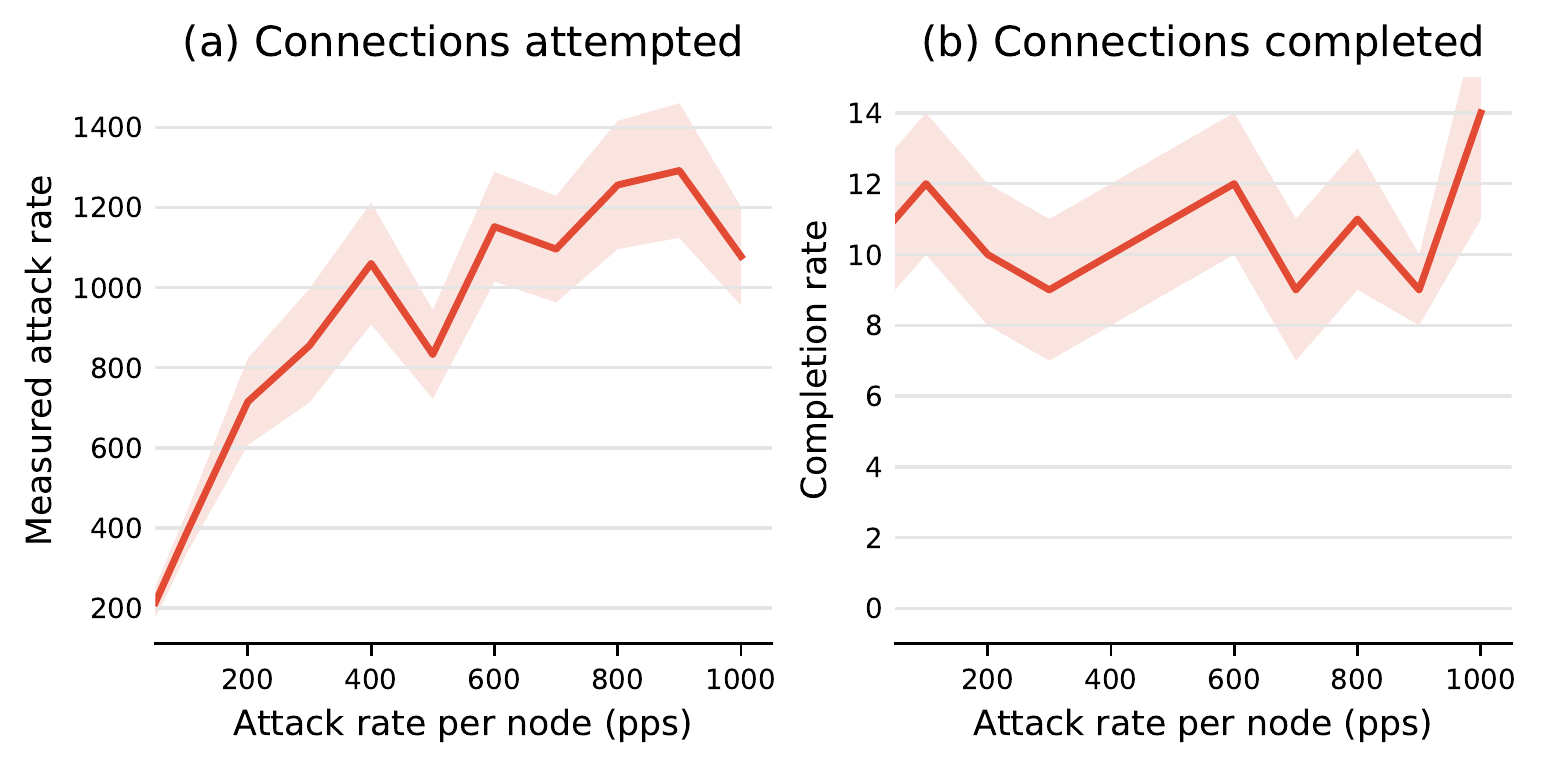}
  \caption{Impact of the puzzles on the attack as the flooding rate is increased.}
    \label{fig:attack_impact_rate}
\end{figure}

\begin{figure}
\centering
  \includegraphics[width=\columnwidth]{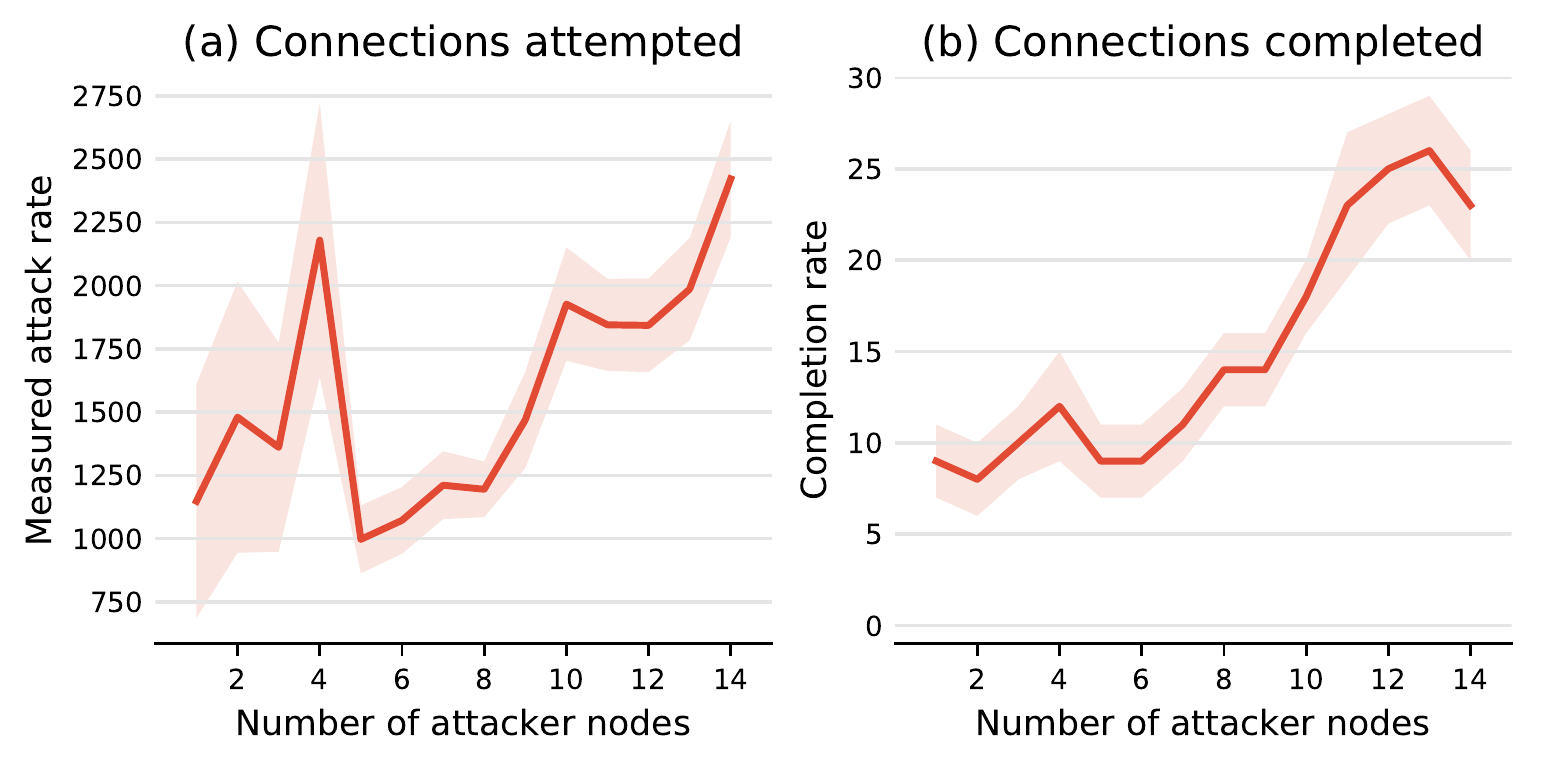}
  \caption{Impact of the puzzles on the attack as the number of participating machines is increased.}
    \label{fig:attack_impact_size}
\end{figure}

In the fourth experiment, we show that TCP puzzles increase the server's tolerance to a botnet and require attackers to increase their botnet's size to deny service. We vary the botnet's size and attack rate and measure the cumulative attack rate as seen by the botnet (referred to as the measured attacker rate) and the server (referred to as the connection completion rate). The connection completion rate is the effective attack rate that actually impacts the server. In the first scenario, we set the number of nodes in the botnet to 5 and vary the sending rate of each node between 100 and 1000 pps. Figure~\ref{fig:attack_impact_rate}(a) shows the measured attack rate and Figure~\ref{fig:attack_impact_rate}(b) shows the rate of completed connections. The results show that the TCP puzzles are capable of rate limiting the effective attack rate. As the per node attack rate increases, the cumulative attack rate increases reaching a peak of $1200$ pps.
However, the effective attack rate is limited to 11 cps regardless of the individual attack rate.

In the second scenario, we vary the number of machines in the botnet while setting the cumulative attack rate to $5,000$ pps; each machine's rate is set at $5,000/(\text{size of botnet})$. 
Figures~\ref{fig:attack_impact_size}(a) and~\ref{fig:attack_impact_size}(b)
  show the measured attack rate and the effective attack rate as the number of machines are varied, respectively. The results show that attackers have to increase the size of their botnets to increase their effective attack rates. As more machines are added, the measured attack rate increases to peak at 2250 pps. 
%
The effective attack rate, although it linearly increases with the increase in the number of attack machines, it only peaks at 25 cps --- one hundredth the measured attack rate.
As opposed to the near-constant rate in the first scenario, the effective attack rate increased in this scenario since more machines have been enlisted in the botnet.
However, this increase does not reflect the increase in resources being committed to the botnet. At the current rate of increase, a botnet has to commit 500 machines to reach an effective attack rate of 5000 cps. 

In conclusion, for the attacker to increase the effective attack rate, it cannot increase the individual rates, it has
to increase the number of machines in the botnet; TCP puzzles at the Nash equilibrium difficulty significantly increases
the cost of a state exhaustion attack.

\subsection{Experiment 5: Adoption of TCP puzzles}
In this experiment, we show that a client solving the TCP puzzles is almost always able to connect to the server
regardless of whether the attacker elects to solve or ignore the puzzles or select a combination thereof. On the other
hand, a client that does not solve puzzles gets erratic service when the attacker is solving the puzzles and almost no
service when the attacker floods the server without solving any puzzles. In this experiment, we use machines that have
not been patched to support TCP puzzles; we test all four scenarios when (1) both the attacker and client do not solve
puzzles (NA,NC), (2) the attacker solves puzzles while the clients do not solve puzzles (SA,NC), 
(3) the clients solve puzzles with the attacker solving the puzzles, and (4) the clients solve puzzles with the attacker
not solving the puzzles. We group scenarios (3) and (4) together and label them as (*A, SC).
Figure~\ref{fig:adoption_impact} shows the percentage of completed connections for all the proposed scenarios. We
observe that a client solving puzzles is not denied service regardless of the attacker's type; this happens because the
attacker, being rate limited when solving puzzles and having its requests ignored when not solving challenge request, is
not able to fill the \texttt{accept} queue of the server. 
On the contrary, a non-solving client faced with a solving attacker experiences a highly variable percentage of
completed connections, reaching 0 at some instances. This happens due to the opportunistic nature of the TCP puzzles
controller (as observed in Experiment 2); the rate limiting impact on the attacker machines can empty slots in the
server's queues thus providing openings for the non-solving client to establish connections. 
However, when faced with an attacker that does not solve the challenges, the non-solving clients are denied service.
This happens because the attacker's vast resources beat the clients' requests for the resources freed by the TCP puzzles
controller. We note that the service promises provided by our puzzles implementation to non-compliant clients are
similar, and sometimes better, than those provided by network capabilities~\cite{Xiaowei2005}, while we almost always
provide service for those who comply.
\begin{figure}
\centering
  \includegraphics[width=0.8\columnwidth]{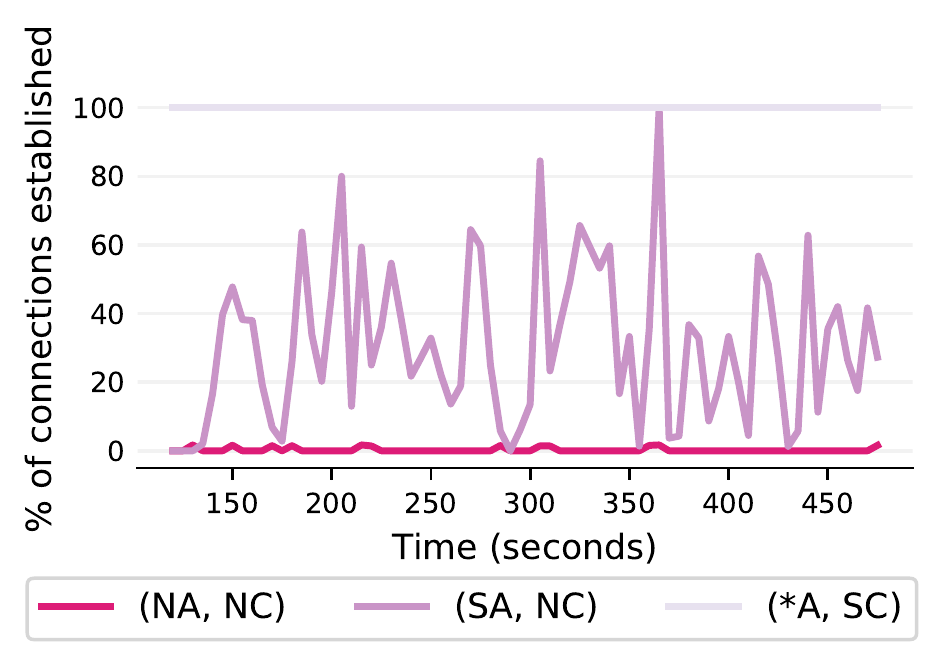}
  \caption{Percentage of established connections when TCP puzzles adoption is not complete -- (*A, SC) captures both (NA, SC) and (SA, SC).}
    \label{fig:adoption_impact}
\end{figure}

\subsection{Experiment 6: Impact on IoT devices}
\begin{table}[t]
\centering
\vspace{2.5em} 
\caption{Performance profile of embedded devices.}
\label{tb:iot_table}
\begin{tabular}{ccc}
\hline
Device        & Average hashing rate & Hashes performed in 400 ms \\ \hline
$D_1$ & 49,617                      & 19,901                                \\ 
$D_2$ & 68,960                      & 26,563                                \\ 
$D_3$ & 70,009                      & 27,987                                \\ 
$D_4$ & 74,201                      & 29,732                                \\ \hline
\end{tabular}
\vspace{-0.5em}
\end{table}
In the last experiment, we explore the impact of the client puzzles on the capabilities 
of IoT devices. We show that the flooding impact of IoT devices can be greatly reduced by virtue
of the support for TCP client puzzles. 
Table~\ref{tb:iot_table} shows the average hashing rate  and the number of
hashes performed in the $400ms$ interval,
of four {\em Raspberry Pi} devices: 
$D_1$ equipped with a 700MHz ARM 11 processor,
$D_2$ equipped with a 1GHz ARM 11 processor,
$D_3$ equipped with a  quad-core 1.2 GHz  ARM Cortex-A53 processor, and
$D_4$ equipped with a quad Core 1.2GHz Broadcom BCM2837 processor. 
The performance profiles show that, while the computational capabilities of these processors still 
enable them the opportunity to connect to a server that is deploying client puzzles, they 
greatly hinder their ability to effectively participate in a distributed TCP connection 
flood attack. Such devices might still be able to send SYN request packets, however, their ability to 
launch a flood of connections is limited. 
This subsequently implies that an attacker recruiting IoT devices as part of 
her bot army needs to employ much more resources to launch an effective 
attack. This achieves our goal of increasing the cost of  
conducting state exhaustion attacks. 

\section{Limitations and Discussion} \label{s:security}
In this section, we discuss the challenges facing the adoption of TCP client puzzles and provide an analysis of their limitations. 

{\em Software adoption}
As showcased by our experiments, there is a great benefit for servers to adopt client puzzles as a mechanism for
tolerance of state exhaustion attacks. By rate limiting users and protecting the server's internal queues,
client puzzles present service providers with a chance to provide continuous
service during state exhaustion attacks.  Our implementation presents the following incentives for ease of adoption.
First, a running server can easily support client puzzles by simply patching its kernel. 
Second, our implementation does not introduce any changes to the normal  operation of the server and only sends
challenges when the queues overflow, the TCP stack remains intact otherwise. Finally, our patch is compatible with earlier versions
of the Linux kernel as long as it has support for the cryptographic operations.

While servers are incentivized to adopt TCP puzzles by the increased tolerance of attacks, clients, on the other hand, benefit from the promise of service even during attacks.
As shown in experiment 5, users that enable support for client puzzles 
are always able to connect to the server.
The server can protect its queues from attackers that do not solve challenges while rate-limiting those that do, thus freeing up resources for the legitimate clients. 
For the users that choose not to adopt the challenges, they still receive full service under regular  load. However,
when under attack, those users will be in contention with the non-solving attackers for the spots that free up in the
queues by the server's opportunistic challenges controller. This scenario is no worse than the case when
no challenges are applied. Note that our theoretical formulation asserts the same outcome, a user that does not adopt
TCP challenges is similar to one that values the server's services at $w = 0$. Such a user prohibits the server from
providing it with protections since it cannot assign an appropriate puzzle difficulty. We believe that the opportunistic
challenges controller offers non-cooperating users service promises similar to those provided by network
capabilities~\cite{Xiaowei2005}, in which a small portion of the channel is shared between non-compatible users and attackers. 
Finally, we note that users can manually elect to disable the use of puzzles through the kernel's \texttt{sysctl} interface.
{\em Solution floods}
When implementing client puzzles as a scheme for resisting state exhaustion attacks, the server must commit computational 
resources to generate and verify puzzles for every incoming connection request. An attacker thus might 
attempt to overwhelm the server with verification tasks by sending a barrage of bogus solutions for which the server must 
perform cryptographic operations to verify their validity. We, however, argue that (1) the overhead of generating and verifying puzzles is negligible and (2) such an attack requires the attacker to commit vast resources that render it infeasible at low costs (thus alleviating the cost imbalance between the attacker and the defender). 

Our game-theoretic formulation takes into consideration the server's computational load and the number of cryptographic operations that it must perform when verifying a puzzle solution; it aims to maximize the clients' load while minimizing the server's load.
This is confirmed by the results discussed in Section~\ref{s:eval:main} and showcased in Figure~\ref{fig:exp2:cpu}. The computational effort employed by the server to generate and verify puzzles is negligible, with CPU utilization remaining below 5\% throughout the entirety of the attack.

Additionally, the server needs only to perform one hashing operation to generate a challenge, and an average of $(1 +\frac{k}{2})$ hashing operations when verifying a solution. We argue that a well-provisioned server can handle such a load while still being able to handle servicing the legitimate clients. 
For example, the server used in our experiments can perform 10.8 million hash operations per second. Thus an attacker aiming to overwhelm such a server would need to send at least 5,400,000 packets per second, each of which contains at least 60 bytes for IP and TCP headers. Therefore, attacking a web service that employs server replication and load balancing requires the attacker to employ additional resources coming at a much higher cost. We acknowledge that no solution can completely eliminate the possibility of a successful attack. However, we believe that our results show the client puzzles can significantly increase the cost of launching such attacks, thus restoring some balance to the uneven attack-defense play-field. 
We also note that, to mitigate such attacks, a server that employs client puzzles can benefit from deploying one or more proxy servers that solely handle the generation
and validation of puzzles, and then delegate the application processing of established connections to backend servers~\cite{waters2004}. 

{\em Replay attacks}
Since the server does not retain state about an incoming connection before receiving a valid challenge solution, an
attacker might capture legitimate clients' solutions and replay them to overflow the server's \texttt{accept} queue. We
note however that for a replayed solution to be validated, the attacker must retain the packet's parameters
(IP addresses, port numbers, and timestamps). Therefore a replayed solution can only be used to occupy one slot in the
server's queue at a time. Additionally, our implementation ensures that puzzles expire after a set
timeout interval. The timeout interval limits an attacker's ability to carry on a replay flood effectively, and thus our
implementation is resistant to such attacks.

{\em Fairness and power considerations}
Our model and implementation of TCP client puzzles use the same difficulty level for all of the users, regardless of whether they are legitimate clients or malicious attackers. We were motivated by two important factors when making this design choice: (1) the requirement to remain stateless until a solution is verified and (2) the difficulty of filtering malicious IP addresses when under attack. This nevertheless raises fairness concerns since the puzzles mechanism is nondiscriminatory and treats every request as a potential malicious request. We note however, that this behavior is only experienced when a victim server is under attack and not during regular operation. Additionally, we believe that our work on client puzzles, as posited in~\cite{feng03case}, can be a catalyst for future exploration of fairness schemes, such as {\em Puzzle Fair Queuing}. We plan to address this issue in our future work. 


Furthermore, we do not currently address the important challenge arising from the presence of a non-uniform mix
between power-limited (e.g., mobile phones, IoT devices) and 
power-endowed (e.g., GPU-enabled desktops) benign devices. Although 
$w_{av}$ in our model attempts to capture the power-mix of clients at design time, it does not keep track of 
clients joining and leaving the system. In fact, this non-uniformity of CPU power is one of the main challenges
facing Bitcoin mining, with mining pools controlling $27\%$ of the market hashing power.
A possible solution would be to switch to memory-based proof-of-work schemes~\cite{abadi2005} that promise more
uniform solution requirements. Another possibility would be to adapt the difficulty of the sent puzzles based 
on the behavior of the observed traffic at the server, thus forming a closed control loop. We plan to investigate
such methods in our future work. 
 

\vspace{-.5em}


\section{Conclusion} \label{s:conclusion}
In this paper, we presented a theoretical formulation and implementation of client puzzles as means for providing
tolerance to state exhaustion attacks. 
We addressed the challenge of selecting puzzle difficulties by modeling the problem as a Stackelberg game where the server is the leader 
and the clients are the followers. 
We obtained the equilibrium solution that 
illustrates a tradeoff between the clients' valuation of the requested services and the server's service capacity. We then 
tackled the lack of puzzle implementations by providing a Linux kernel patch and evaluating its performance on the DETER 
testbed. 
Our results show that client puzzles are an effective mechanism that can be added to our arsenal of 
defenses to increase the resilience to multi-vectored DDoS attacks and restore the balance to the
attack-defense play field. 


\bibliographystyle{ACM-Reference-Format}
\bibliography{references}

\appendix
\section{Proof of Theorem 1}\label{app:proof}
In order to analytically solve for the equilibrium solution of the game, we follow an approach similar to that
in~\cite{basar2002}. We start by noting that the Nash Equilibrium solution of the users' game is not affected if we add
the quantity
\[
  \sum_{j\neq i} \parans{w_j \log(1+x_j) - k\times 2^{m-1}x_j}
\]
to each users' utility function. Therefore we can now build a strategically equivalent game where each user's utility
function is 
\begin{equation}
  H\parans{x_1, \ldots, x_N, p} = \sum_{i=1}^N w_i\log (1+x_i) - k\times 2^{m-1} \bar{x} - \frac{1}{\mu - \bar{x}}
  \label{eq:sol_user_utility}
\end{equation}
Now looking at the Hessian matrix of $H$ we get 
\[
  \begin{aligned}
  \mathbf{H}_{ii} = \frac{\partial^2 H}{\partial x_i^2} = - \frac{w_i}{(1+x_i)^2} - \frac{2}{(\mu - \bar{x})^2} < 0, &
  \quad \forall{i} \\
  \mathbf{H}_{ij} = \frac{\partial^2 H}{\partial x_ix_j} = -\frac{2}{(\mu - \bar{x})^3} < 0, & \quad \forall{i,j}, \;i\neq
  j 
\end{aligned}
\]
Therefore $\mathbf{H}$ is negative-definite and thus $H$ is strictly concave for $0 \leq \bar{x} < \mu$. Additionally,
since $H \rightarrow -\infty$ as $\bar{x} \rightarrow \mu$, we can conclude the that optimization problem 
\[
  \underset{x_i \geq 0:\forall{i}, \bar{x} < \mu}{\max} H\parans{x_1, x_2, \ldots, x_N, p}
\]
admits a unique solution $\mathbf{x}^*=\set{x^*_1, \ldots, x^*_N}$ in the interval $0 \leq \bar{x} < \mu$ which corresponds 
to the Nash Equilibrium strategies to the users' game as defined in Equation~\eqref{eq:app_user_utility}.
We obtain the solution strategies by solving the first order condition of $H$ where for $i \in \set{1, \ldots, N}$
\[
  \frac{\partial H}{\partial x_i} \parans{x^*_1, \ldots, x^*_N, p} = 0
\]
which translates to 
\begin{equation}
  \frac{w_i}{1+x^*_i} - k \times 2^{m-1} - \frac{1}{\mu -\bar{x}^*} = 0, \; \forall{i}
  \label{eq:first_order_cond}
\end{equation}

\noindent Let $y_i = 1 + x_i$, $\bar{y} = \sum_{i=1}^N y_i = N + \bar{x}$, and $\bar{w} = \sum_{i=1}^N w_j$, from which
we obtain 
\[
  \frac{w_i}{y_i} = \frac{w_j}{y_j},\quad \forall{i,j \in \set{1,\ldots,N}}
\]
or equivalently 
\[
  y_i = \frac{w_i}{w_j}y_j, \quad \forall{i,j \in \set{1,\ldots,N}}
\]
We can then rewrite $\bar{y}$ as 
\[
  \bar{y} = \sum_{i=1}^{N} y_i = \sum_{i=1}^N \frac{w_i}{w_j} y_j = \frac{\bar{w}}{w_j} y_j
\]
and thus we can express \eqref{eq:first_order_cond} in terms of $\bar{y}$ as 
\begin{equation}
  \tilde{L}(\bar{y}) = \frac{\bar{w}}{\bar{y}} - k \times 2^{m-1} - \frac{1}{(\mu + N - \bar{y})^2} = 0
  \label{eq:first_order_ybar}
\end{equation}
We can thus turn our attention to solving Equation~\eqref{eq:first_order_ybar} for $N \leq \bar{y} < \mu + N$. 
Since $\frac{\partial \tilde{L}}{\partial \bar{y}} = - \frac{\bar{w}}{\bar{y}^2} - \frac{2}{(\mu + N -\bar{y})^2} < 0$, 
$\tilde{L}$ is strictly decreasing. Additionally, $\tilde{L}(\bar{y}) \rightarrow -\infty$ as $\bar{y} \rightarrow \mu +
N$. We therefore need $\tilde{L}(N)$ to be non-negative so that $\tilde{L}(\bar{y})$ would admit a solution in 
the interval $N \leq \bar{y} < \mu + N$, which translates to 
\[
  \tilde{L}(N) = \frac{\bar{w}}{N} - k \times 2^{m-1} - \frac{1}{\mu^2} > 0
\]
or equivalently 
\begin{equation}
  k \times 2^{m-1} < \frac{\bar{w}}{N} - \frac{1}{\mu^2} := \hat{r}
  \label{eq:cond_on_pmax}
\end{equation}
We can see $\hat{r}$ as the maximum possible difficulty that the service provider can select while guaranteeing that a
solution for the clients' game exists. We also notice that if the provider had infinite resource, i.e., $\mu \rightarrow
\infty$, $\hat{r} \rightarrow \frac{\bar{w}}{N}$ which suggests that a client should not be charged a price higher than
the average user valuation of the provider's services. 

Furthermore, it is beneficial for the service provider to ensure that all clients participate in the game, i.e., that 
$x_i > 0$ for all $i \in \set{1,2,\ldots, N}$. This therefore translates to the conditions on $\bar{y}$
\begin{equation}
  \bar{y}  > \frac{\bar{w}}{w_i} \quad \forall i
  \label{eq:condition_on_y}
\end{equation}

Now let $\bar{y}(k,m)$ be a solution to Equation~\eqref{eq:first_order_ybar} that satisfies
condition~\eqref{eq:condition_on_y} and where $(k,m)$ satisfy condition~\eqref{eq:cond_on_pmax}, and let $\bar{x}(k,m)$ be
the corresponding value of $\bar{x}$. We turn to the provider's problem of finding the optimal pricing $p^* = (k^*,
m^*)$ that maximizes 
\begin{equation}
  I(p) := \parans{k\times 2^{m-1} -2 - \frac{k}{2}}\bar{x}(k,m)
  \label{eq:optimization_unsolvable}
\end{equation}

In order to obtain an analytical solution to the optimization problem in Equation \eqref{eq:optimization_unsolvable} we
make use of the following approximation. We solve for the pricing $\tilde{p}(\tilde{k},\tilde{m})$ that maximizes
\begin{equation}
  \tilde{I}(p):= \parans{k\times 2^{m-1}}\bar{x}(k,m)
  \label{eq:reduced_optimization}
\end{equation}

\begin{lem}
$|I(p^*) - \tilde{I}(\tilde{p})| < c$ for some constant $c > 0$, 
where $p^*$ and $\tilde{p}$ are the solutions that
maximize $I$ and $\tilde{I}$, respectively.
\end{lem}
\begin{proof}
  Let $p^* = (k^*, m^*)$ and $\tilde{p} = (\tilde{k}, \tilde{m})$ be the prices that maximize $I(p)$ and $\tilde{I}(p)$,
  respectively. We therefore have that 
  \[
	\tilde{k}\times 2^{\tilde{m}-1} \bar{x}(\tilde{k}, \tilde{m}) \geq k\times 2^{m-1} \bar{x}(k,m), \; \forall k,m
  \]
  Let $p'=(k',m')$ be a price with minimum $0 < k' \leq \tilde{k}$ such that $k'\times 2^{m'-1} = \tilde{k}\times 2^{\tilde{m}-1}$ and
  $I(k',m') \geq I(\tilde{k}, \tilde{m})$. We can therefore write 
  \[
	I(p') \geq I(k,m) - (\frac{k'}{2} + 2)\bar{x}(k',m'), \; \forall{k,m}
  \]
  and since $I(p^*) \geq I(p) \; \forall{p}$ we can therefore conclude that 
  \[
	|I(p^*) - I(p')| \leq (\frac{k'}{2} +2)\bar{x}(k',m') < \parans{(\frac{k'}{2} + 2)\mu} := c
  \]
  and since $\bar{x}$ only depends on $k\times 2^{m-1}$ and not on the individual values of $k$ and $m$, $p'$ also
  maximizes Equation~\eqref{eq:reduced_optimization} and thus solving for $p'$ brings us within a constant $c$ of $p^*$, the maximum of $I$. 
\end{proof}

We can now proceed with finding a solution for Equation~\eqref{eq:reduced_optimization} following the approach presented
in~\cite{basar2002}. By using the one-to-one correspondance between $k\times 2^{m-1}$ and $\bar{y}$ (and thus $\bar{x}$)
presented in Equation~\eqref{eq:first_order_ybar}, we can substitute $\bar{y}$ in~\eqref{eq:reduced_optimization} and then
compute $p^*$ using the solution to the obtained equation. We then write the equivalent problem as finding $\bar{y}^*$
such that 
\begin{equation}
  \bar{y}^* = \underset{N < \bar{y} < N+\mu}{\argmax} \parans{\frac{\bar{w}}{\bar{y}} - \frac{1}{(\mu+N-\bar{y})^2}}(\bar{y}
  - N)
  \label{eq:equiv_optimization}
\end{equation}

We define $G(\bar{y}) := \parans{\frac{\bar{w}}{\bar{y}} - \frac{1}{(\mu+N-\bar{y})^2}}(\bar{y} - N)$. It is easy to see
that $\frac{\partial^2 G}{\partial \bar{y}^2} < 0$ and thus $G(\bar{y})$ is strictly concave. Additionally, $G(N + \mu)
\rightarrow -\infty$, we can thus conclude that $G(\bar{y})$ admits a unique maximum in the open interval $(N, N+\mu)$. 
We can then solve the first order condition 
\begin{equation}
  \frac{\partial G(\bar{y})}{\partial \bar{y}} := \frac{\bar{w}N}{\bar{y}^2} - \frac{\mu+\bar{y}-N}{(\mu+N-\bar{y})^3} = 0
  \label{eq:final_optimization}
\end{equation}

Obtaining a closed form solution for $\bar{y}^*$ is not possible for finite $N$. Therefore we solve Equation \eqref{eq:final_optimization} asymptotically (i.e., as $N \rightarrow \infty$) as proposed in~\cite{basar2002}. For that, we make the following
assumptions. (1) We assume that the average user preference $w_{av}(N) = \frac{\bar{w}}{N}$ has a well defined
limit $w_{av}$ as $N \rightarrow \infty$.
(2) We assume that as the number of users grows larger, the service provider can always service a fraction of its users,
even if that fraction is small. In other words, we assume that $\underset{N\rightarrow \infty}{\lim} \frac{\mu}{N} = \alpha$ for
some $\alpha > 0$. For convenience, we rewrite Equation~\eqref{eq:final_optimization} in terms of $x_{av}(N) =
\frac{\bar{x}}{N}$ and $w_{av} (N)$ as $N \rightarrow \infty$ as 
\begin{equation}
  \frac{w_{av}}{(1+x_{av}(N))^2} = \frac{\alpha + x_{av}(N)}{(\alpha-x_{av}(N))^3N^2}
  \label{eq:refined}
\end{equation}

Equation~\eqref{eq:refined} possesses a solution for $x_{av}(N)$ iff,
\[
\underset{N \rightarrow \infty}\lim \left( \alpha - x_{av}(N) \right)^3N^2 = \gamma
\]
for some $\gamma > 0$. We thus substitute back in Equation~\eqref{eq:first_order_ybar} and obtain the solution 
\begin{equation} \label{eq:solution_long}
k^* \times 2^{m^*-1} \sim \frac{w_{av}}{\alpha + 1} + \frac{2 \alpha - 1}{\gamma^{\frac{2}{3}}N^{\frac{2}{3}}}    
\end{equation}
where $f \sim g$ denotes the fact that 
$\underset{N \rightarrow \infty}{\lim} \frac{f}{g} = 1$.

Since we are considering the asymptotic solution, we restrict our attention to the first order
term of the solution in Equation~\eqref{eq:solution_long} and thus obtain our desired form
\begin{equation} \label{eq:solution_final}
    \boxed{k^* \times 2^{m^*-1} = \frac{w_{av}}{\alpha + 1}}
\end{equation}
In fact, as shown in~\cite{basar2002}, Equation~\eqref{eq:solution_final} corresponds to the solution of the same problem when ignoring the service delay at the server. Since SYN and connection flood attacks target the TCP protocol and not the application layer service, it is convenient for the purposes of this paper to only consider the first order term of Equation~\eqref{eq:solution_long}, thus completing the proof. \hfill $\blacksquare$

\appendix
\section{Experiment Topology and machines}\label{app:topology}
Figure~\ref{fig:topology} shows the topology of our experiment setup. All the hosts are physical machines. The following is a list of the hardware specifications of the machines used in the setup. The nomenclature pcxxx and 
bpcxxx refers to the specific hardware models provided by the DETER testbed. Our client and attacker machines use a combination of those models while our server is deployed on a more capable HP server having dual Intel Xeon hex-core processors running at 2.2 Ghz, 24 GB of RAM, 1 TB of storage, and a 10-Gigabit network interface. 
\begin{figure}
\center
	\includegraphics[width=1\columnwidth]{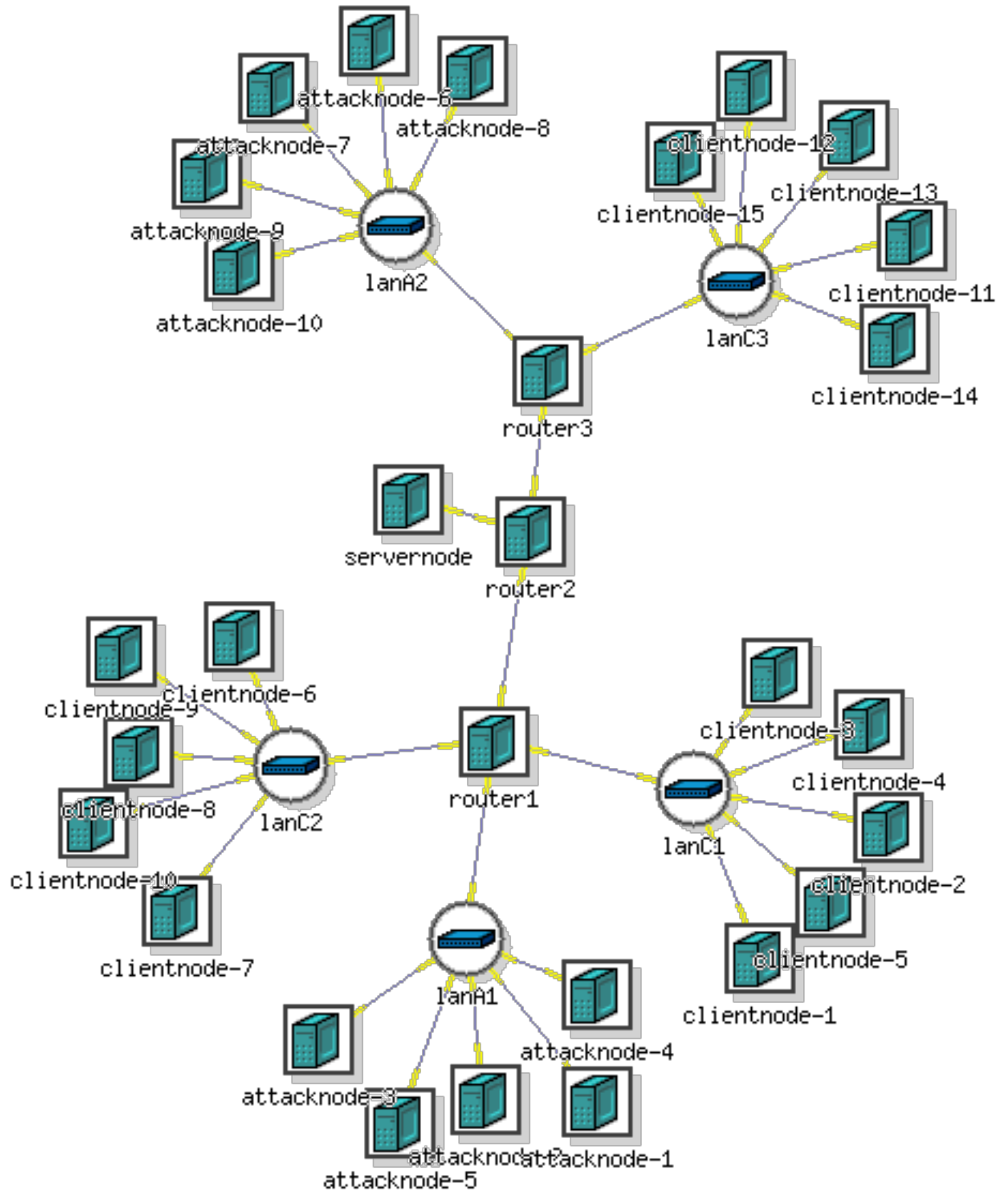}
	\caption{Scenario Topology}
	\label{fig:topology}
\end{figure}

pc3000 and bpc3000 have the following features:
\begin{itemize}
\item Dell PowerEdge 1850 Chassis.
\item Dual 3Ghz Intel Xeon processors.
\item 2 GB of RAM
\item One 36Gb 15k RPM SCSI drive.
\item 4 Intel Gigabit experimental network ports.
\item 1 Intel Gigabit experimental network port.
\end{itemize}

The pc2133 and bpc2133 machines have the following features:
\begin{itemize}
	\item Dell PowerEdge 860 Chasis
\item One Intel(R) Xeon(R) CPU X3210 quad core processor running at 2.13 Ghz
\item 4GB of RAM
\item One 250Gb SATA Disk Drive
\item One Dual port PCI-X Intel Gigabit Ethernet card for the control network.
\item One Quad port PCIe Intel Gigabit Ethernet card for experimental network.
\end{itemize}

High Density SuperMicro MicroCloud Chassis that fits 8 nodes in 3u of rack space have the following features:
\begin{itemize}
	\item One Intel(R) Xeon(R) E3-1260L quad-core processor running at 2.4 Ghz
Intel VT-x and VT-d support
\item 16GB of RAM
\item One 250Gb SATA Western Digital RE4 Disk Drive
\item 5 experimental interfaces
\item One Dual port PCIe Intel Gigabit Ethernet card for the control network and an experimental port
\item One Quad port PCIe Intel Gigabit Ethernet card for experimental network
\end{itemize}

HP Proliant DL360 G8 Server have the following features:
\begin{itemize}
	\item Dual Intel(R) Xeon(R) hexa-core processors running at 2.2 Ghz with 15MB cache Intel VT-x support
\item 24GB of RAM
\item One 1Tb SATA HP Proliant Disk Drive 7.2k rpm G8 (boot priority)
\item One 240Gb SATA HP Proliant Solid State Drive G8
\item Two experimental interfaces:
\item One Dual port PCIe Intel Ten Gigabit Ethernet card for experimental ports
\item One Quad port PCIe Intel Gigabit Ethernet card, presently with one port wired to the control network
\end{itemize}

The bpc2800 machines have the following features:
\begin{itemize}
	\item Sun Microsystems Sun Fire V60 Chassis
\item One Intel(R) Xeon(R) CPU dual core processor running at 2.8 GHz
\item 2 GB of RAM
\item One 36 GB SCSI Disk Drive
\item Two Dual port PCI-X Intel Gigabit Ethernet cards, 1 port for control network and 3 ports for experimental network
\item One Single port PCI-X Intel Gigabit Ethernet card for experimental network
\end{itemize}

%

Finally, for our IoT experiment, we used four Raspbery Pi boards ranging over the released model revisions. Specifically we used (1) a Raspberry Pi Model B, revision 2.0, (2) a Raspberry Pi Zero, (3) a Raspbery Pi 2 Model B v 1.1, and (4) a Raspberry Pi 3 Model B v 1.2. We deployed the Linux kernel 4.9.65 on all four boards and used the kernel's cryptographic API to profile their performance. 

\end{document}